\documentclass{article}

\usepackage{float}
\usepackage{amsmath}
\usepackage{amsthm}
\usepackage{amsfonts}
\usepackage{stmaryrd}
\usepackage{amssymb}
\usepackage{mathtools}
\usepackage{xfrac}
\usepackage{bm}
\usepackage[svgnames]{xcolor}
\usepackage{comment}
\usepackage{authblk}
\usepackage{dsfont}
\usepackage{booktabs}
\usepackage{cprotect}
\usepackage{todonotes}
\usepackage{tikz-cd}
\usepackage[shortlabels]{enumitem}
\usepackage{adjustbox}
\usepackage[pass]{geometry}
\usepackage{pseudo}
\pseudoset{
  compact,
  label=\scriptsize\sffamily\color{gray}\arabic*,
  ctfont=\sffamily\color{gray},
  ct-left=\ctfont{\# },
  ct-right=,
  indent-length=3ex,
  line-height=1.2,
}

\usepackage{tcolorbox} 
\tcbuselibrary{skins,theorems}

\usepackage[plainpages=false,pdfpagelabels,colorlinks=true,allcolors=blue!80!black,hypertexnames=false,pdftex=true,unicode]{hyperref}

\usepackage[style=numeric, sorting=nyt, isbn=false, uniquename=false, giveninits=true, url=false, maxnames=99]{biblatex}
\newenvironment{algovals}{\begin{list}{\scriptsize$\bullet$}{\setlength{\topsep}{0pt}\setlength{\parskip}{0pt}\setlength{\itemsep}{0pt}\setlength{\parsep}{0pt}}}{\end{list}}

\DeclareMathOperator{\lm}{lm}
\DeclareMathOperator{\lc}{lc}
\DeclareMathOperator{\lt}{lt}

\DeclareMathOperator{\ann}{Ann}

\DeclareMathOperator{\supp}{supp}
\DeclareMathOperator{\Span}{Span}
\DeclareMathOperator{\LRem}{LRem}
\DeclareMathOperator{\RRem}{RRem}
\DeclareMathOperator{\Reduce}{Reduce}
\DeclareMathOperator{\ai}{Ai}
\DeclareMathOperator{\grevlex}{Grevlex}

\newcommand{\monxx}{M_\xx}                   
\newcommand{\monxxr}{M_{\xx,r}}              
\newcommand{\montxxr}{M_{t,\xx,r}}           

\newcommand{\scalp}[2]{\langle#1,#2\rangle}

\newcommand{\res}{\operatorname{res}}
\newcommand{\proj}{\operatorname{pr}}
\newcommand{\ind}{\operatorname{ind}}

\def\st{\ \middle|\ }

\newtcbtheorem{alg}{Algorithm}{pseudo/booktabs, float}{algo}

\newcommand{\msrp}{\bQ\langle\langle\pp\rangle\rangle}

\def\ud{\mathrm{d}}
\def\xx{\mathbf{x}}
\def\pp{\mathbf{p}}
\def\rr{\mathbf{r}}
\def\ss{\mathbf{s}}
\def\uu{\mathbf{u}}

\def\dd{\bm{\partial}}
\def\aalpha{{\bm\alpha}}
\def\bbeta{{\bm\beta}}
\def\ggamma{{\bm\gamma}}
\def\ddelta{{\bm\delta}}

\newtheorem{theorem}{Theorem}
\newtheorem{proposition}[theorem]{Proposition}
\newtheorem{corollary}[theorem]{Corollary}
\newtheorem{lemma}[theorem]{Lemma}

\theoremstyle{definition}
\newtheorem{definition}[theorem]{Definition}
\newtheorem{hypothesis}[theorem]{Hypothesis}

\theoremstyle{remark}
\newtheorem{remark}[theorem]{Remark}
\newtheorem{example}[theorem]{Example}

\newcommand{\Irr}{E}

\newcommand{\bF}{\mathbb F}
\newcommand{\bK}{\mathbb K}
\newcommand{\bN}{\mathbb N}
\newcommand{\bQ}{\mathbb Q}
\newcommand{\bZ}{\mathbb Z}
\newcommand{\bR}{\mathbb R}
\newcommand{\bC}{\mathbb C}
\newcommand{\bL}{\mathbb L}

\newcommand{\cF}{\mathcal F}
\newcommand{\cL}{\mathcal L}

\title{Faster multivariate integration in D-modules}

\usepackage{authblk}
\author[1]{Hadrien Brochet}
\author[1]{Frédéric Chyzak}
\author[1]{Pierre Lairez}
\affil[1]{Inria, 91120, Palaiseau, France}
\date{January 28, 2026}

\begin{document}

\maketitle

\begin{abstract}
We present a new algorithm for solving the reduction problem in the context of holonomic integrals, which in turn provides an approach to integration with parameters. Our method extends the Griffiths--Dwork reduction technique to holonomic systems and is implemented in Julia. While not yet outperforming creative telescoping in D-finite cases, it enhances computational capabilities within the holonomic framework. As an application, we derive a previously unattainable differential equation for the generating series of 8-regular graphs.
\end{abstract}

\section{Introduction}

Symbolic integration is a fundamental problem in computer algebra, with deep connections to combinatorics, number theory, and mathematical physics.
In the vast landscape of integrable functions, holonomic functions---those multivariate functions satisfying sufficiently many independent linear partial differential equations with polynomial coefficients---form a particularly rich and structured class. This class includes many special functions of mathematical and physical interest, such as exponential functions, logarithms, polylogarithms, elliptic integrals, and various hypergeometric functions.
In the realm of univariate functions,
the classical problem of integrating elementary functions in terms of elementary functions does not always have a solution---as governed by Liouville's theorem.
Holonomic functions, in contrast, exhibit a different structure that is closed under integration,
including under definite integration depending on parameters.
In the holonomic setting (and its variants), symbolic integration revolves around two problems: \emph{integration with a parameter} and \emph{reduction}.

\paragraph{Integration with a parameter.}
Given a function~$f(t, x_1,\dotsc,x_n)$ of ${1+n}$ variables satisfying a suitable system of linear partial differential equations with polynomial coefficients, we aim to compute a linear differential equation satisfied by the integral
\[ I(t) = \int_D f(t, x_1,\dotsc,x_n)\ud x_1 \dotsb \ud x_n. \]
To obtain an algebraic formulation of the problem and eliminate the analytic aspects, we need some simplifying assumptions.
We consider a function space~$M$ containing~$f$ and satisfying:
\begin{enumerate}[(a)]
  \item\label{item:1} $M$ is closed under differentiation by~$t$ and the~$x_i$, and by multiplication by~$t$ and the~$x_i$;
  \item $g\mapsto \int_D g(x_1,\dots,x_n) \ud x_1 \dots \ud x_n$ is well defined on~$M$ and commutes with differentiation with respect to~$t$ and multiplication by~$t$;
  \item For any~$g\in M$ and any~$1\leq i\leq n$, $\int_D \frac{\partial}{\partial x_i} g(x_1,\dots,x_n) \ud x_1 \dots \ud x_n = 0$.
\end{enumerate}

Then the problem reduces to finding polynomials~$p_0(t), \dotsc,p_r(t)$, with~$p_r$ nonzero, and functions~$g_1,\dotsc,g_n \in M$ such that
\begin{equation}\label{eq:x-free}
  p_0(t) f + p_1(t) \frac{\partial f}{\partial t} + \dotsb + p_r(t) \frac{\partial^r f}{\partial t^r} = \frac{\partial g_1}{\partial x_1} + \dotsb + \frac{\partial g_n}{\partial x_n}.
\end{equation}
Indeed, after integrating both sides over~$D$, the hypotheses above imply that
\begin{equation}\label{eq:target-rel}
  p_0(t) I + p_1(t) \frac{\partial I}{\partial t} + \dotsb + p_r(t) \frac{\partial^r I}{\partial t^r} = 0,
\end{equation}
which is the kind of relation we aim to compute.

To take one further step towards algebra,
we introduce the Weyl algebra~$W_{t,\xx}$,
which is the non-commutative algebra generated by~$t$, $x_1,\dotsc, x_n$, $\partial_t$, $\partial_{x_1}, \dotsc, \partial_{x_n}$, and the usual relations~$u v = vu$, $\partial_u \partial_v = \partial_v \partial_u$ and $\partial_u u = u \partial_u + 1$, for any distinct~$u, v\in \left\{ t, x_1,\dotsc,x_n \right\}$.
Hypothesis~\ref{item:1} translates to the statement that~$M$ is a left module over~$W_{t,\xx}$.
To ensure that \eqref{eq:x-free}~has a solution, we require more specifically that:
\begin{enumerate}[(a), resume]
  \item $M$ is a holonomic~$W_{t,\xx}$-module.
\end{enumerate}
The concept of holonomicity embodies the idea of a function satisfying “sufficiently many independent linear PDEs”, see Section~\ref{sec:holonomic-modules} for more details.

From the algorithmic point of view,
we want an algorithm that takes as input a description of~$M$ as a~$W_{t,\xx}$-module, with generators and relations,
and computes polynomials~$p_0(t),\dotsc,p_r(t)$ such that \eqref{eq:x-free}~holds for some $g_1,\dotsc,g_n$ in~$M$, which we usually do not need to compute.

\paragraph{The reduction problem.}
In the absence of a parameter, integrals are constants and holonomic methods do not directly compute them. However, holonomic methods can be used to find relations between integrals. Consider the Weyl algebra~$W_{\xx}$ in the variables~$x_1,\dotsc,x_n$ and a holonomic~$W_{\xx}$-module~$M$.
As above, we algebraically interpret integration as a linear map on~$M$
that vanishes on the subspace~$\partial_{x_1} M + \dotsb + \partial_{x_n} M$, which we denote simply~$\dd M$ below.
So, finding a relation between the integrals of functions~$f_1,\dotsc,f_r \in M$ means finding constants~$c_1,\dotsc,c_r$, with~$c_r \neq 0$, such that
\[ c_1 f_1 + \dotsb+ c_r f_r \in \dd M. \]
Computing normal forms in the quotient~$M/\dd M$ is one way to find this sort of relations.
Owing to the assumption of the holonomicity of~$M$,
a classical result is that this quotient is finite-dimensional over the base field
\parencite[Theorem~6.1 of Chapter~1, combined with the example that precedes it, for~$p=2n$]{Bjork-1979-RDO}.
It is now well understood that the integration problem may be tackled
through the reduction problem \cite{BostanChenChyzakLi_2010,BostanChenChyzakLiXin_2013,ChenHuangKauersLi_2015,ChenKauersKoutschan_2016,ChenHoeijKauersKoutschan_2018,BostanChyzakLairezSalvy-2018-GHR,Hoeven_2021},
in a way that departs from algorithms based on earlier approaches,
namely by a linear ansatz \cite{WilfZeilberger_1992,Koutschan_2010a}
or by following Zeilberger's method \cite{Zeilberger_1990a,Chyzak_2000}.

\paragraph{Contributions.}

In the holonomic context, we propose a new algorithm for the reduction problem (Section~\ref{sec:reduction}).
In some aspects, this is a generalization to holonomic systems of the Griffiths--Dwork reduction method for homogeneous rational functions \cite{Griffiths-1969-PCR,BostanLairezSalvy_2013,Lairez-2016-CPR} (Section~\ref{subsec:GD}).
This algorithm yields a new algorithm for integration with a parameter (Section~\ref{sec:CTbyRed}).
We provide an implementation of our algorithm in Julia~\cite{bezanson2017} (Section~\ref{sec:implementation}).
Although the new algorithm is still not on par with the best implementations following Zeilberger's approach and generalizations on their home turf (i.e., D-finite functions, see below),
it improves the state of the art in the holonomic context.
As an application, we were able to compute a differential equation for the generating series of 8-regular graphs,
for which Zeilberger's approach is theoretically not suited,
and which was also previously unattainable by dedicated methods (Section~\ref{sec:kreg}).

\paragraph{D-finiteness \emph{versus} holonomicity.}
The \emph{creative telescoping} approach to symbolic integration (see \cite{Chyzak-2014-ACT} for a review) relies on \emph{D-finiteness} instead of holonomicity.
Instead of working with the Weyl algebra~$W_{t,\xx}$ and with holonomic $W_{t,\xx}$-modules,
this approach considers the rational Weyl algebra $W_{t,\xx}(t, \xx)$,
which is the Weyl algebra extended with rational functions in~$t$ and~$\xx$,
and $W_{t,\xx}(t,\xx)$-modules that are finite-dimensional over~$\bK(t,\xx)$, where~$\bK$ is the base field.
This nuance has deep concrete implications.

Expressivity is an argument in favor of holonomicity: we can express a wider class of integrals with holonomicity than with D-finiteness.
Integrals over semialgebraic sets are a prominent example \cite{Oaku_2013}. In general, it is always possible to construct a D-finite module from a holonomic module~$M$:
it is enough to consider the localization $\bK(t,\xx)\otimes_{\bK[t,\xx]} M$, but this operation may lose important information.
For example, it is possible that $\bK(t,\xx)\otimes_{\bK[t,\xx]} M = 0$, as it happens when enumerating $k$-regular graphs (see Section~\ref{sec:kreg}),
making it impossible to apply any creative-telescoping algorithm over rational functions in a relevant~way.

As for efficiency, approaches based on D-finiteness (implemented in Mathematica \cite{Koutschan_2010a,Koutschan_2010}, Maple \cite{Chyzak_2000,BostanChyzakLairezSalvy-2018-GHR}, and SageMath \cite{KauersJaroschekJohansson-2015-OPS}) are far superior, when they apply,
to those based on holonomicity \cite{OakuTakayama_2001} (implemented in Macaulay2, Singular, and Risa/Asir).
Understanding and bridging this gap to achieve both efficiency and expressivity remains a significant challenge, and this paper represents a step forward in that direction.

In this paper, we consider a mixed approach by using modules over~$W_{t,\xx}(t)$,
the Weyl algebra extended with rational functions in~$t$ only: this amounts to considering holonomicity with respect to the~$\xx$ and D-finiteness with respect to the parameter~$t$.
This enables the use of reductions over the base field~$\bK(t)$ to compute integrals with a parameter.

\paragraph{Related work.}
In the context of D-finiteness, the problem of integration with parameters is addressed by algorithms of \textcite{Chyzak_2000,Koutschan_2010a}.
A more recent research line addresses the integration problem by solving first the reduction problem \parencite{BostanChenChyzakLi_2010,BostanChenChyzakLiXin_2013,ChenHuangKauersLi_2015,ChenKauersKoutschan_2016, ChenHoeijKauersKoutschan_2018,BostanChyzakLairezSalvy-2018-GHR,Hoeven_2021,ChenDuKauers_2023}.
In a holonomic context, the integration problem and the reduction problem have been addressed by \textcite{Takayama_1990,OakuTakayama_2001},
without making efficiency their main goal.
In this work,
we forgo the minimality of the order of the output relation~\eqref{eq:target-rel}
in order to bypass the machinery related to $b$-functions.
Already Takayama's algorithm~\cite{Takayama_1990} made a similar compromise.
However, in comparison to this algorithm,
we leverage a Gröbner basis technique to obtain a first reduction, which has a lot of structure and can be computed efficiently
but is not enough to detect all relations between integrals.
This is completed by another reduction, more direct and less structured, which takes place in a lower-dimensional space than what would arise with Takayama's algorithm, thanks to the first reduction.

In the case of integrals of the form \[
  \int a(x_1,\dots,x_n) \exp(f(x_1,\dots,x_n)) \ud x_1 \dotsb \ud x_n,
  \]
where $a(x_1,\dots,x_n)$~is a polynomial and $f(x_1,\dots,x_n)$~a homogeneous polynomial,
the first reduction echoes the Griffiths--Dwork reduction \cite{Griffiths-1969-PCR,BostanLairezSalvy_2013},
while the second reduction echoes Lairez's reduction algorithm in \cite{Lairez-2016-CPR}.
In the similar context of rational integrals,
which are of great importance in the computation of Feynman integrals, the reduction is addressed by \textcite{Laporta_2000}.

In the context of the combinatorics of $k$-regular graphs,
first algorithms for computing linear differential equations
satisfied by their counting generating functions
were developed by \textcite{ChyzakMishnaSalvy-2005-ESP}.
This was following works by the combinatorialist Gessel in the 1980s,
who introduced a representation of the generating functions
as a scalar product in the theory of symmetric functions.
A faster method was very recently introduced by \textcite{ChyzakMishna-2024-DES},
based on the same scalar-product representation
but following an approach reminiscent of reductions.
This was the starting point of our interest,
making us rethink the representation
to have the algorithms of the present work apply directly to the problem.

\section{Computing with Weyl algebras}
\label{sec:weyl-alg}

\subsection{Weyl algebras}

Let $\bK$ be a field of characteristic zero,
typically $\bQ$ or~$\bQ(t)$.
Let $W_{\xx}$ denote the \emph{$n$th Weyl algebra} $\bK[\xx]\langle \dd_\xx \rangle$
with generators $\xx=(x_1,\dots,x_n)$ and $\dd_\xx = (\partial_{1},\dots,\partial_{n})$,
and relations $\partial_{i} x_i = x_i \partial_{i} + 1$, $x_ix_j = x_jx_i$, $\partial_{i} \partial_{j} = \partial_{j} \partial_{i}$ and~$x_i \partial_j = \partial_j x_i$
whenever~$i\neq j$.
We refer to~\cite[Chapters~1--10]{Coutinho_1995} for a complete introduction to these algebras
covering most needs of the present article,
or to \parencite[Chapter~5]{Borel-1987-ADM} for a denser alternative.
We often need to highlight one variable with a specific role,
in which case we use the name~$t$ for the distinguished variable.
Correspondingly, we will write $W_{t,\xx}$
for the $(1+n)$th Weyl algebra, and we will write $W_t$ for the special case~$n=0$.
We also define $W_{t,\xx}(t)$ as the algebra $\bK(t) \otimes_{\bK[t]} W_{t,\xx}$ where the variable $t$ is rational and the variables $\xx$ are polynomial.
For non-zero $r\in\bN$, we also consider Cartesian powers of these algebras,
$W_\xx^r$, $W_{\xx}(t)^r$, etc.,
which we view as modules over $W_\xx$ or $W_{\xx}(t)$, as relevant.
Each element of the module~$W_\xx^r$ decomposes uniquely in the basis of $\bK$-vector space
\begin{equation*}
  \monxxr = \{\xx^\aalpha \dd_\xx^\bbeta e_i \mid \aalpha,\bbeta\in \bN^n, \ i\in\{1,\dots,r\}\},
\end{equation*}
where $e_1,\dots,e_r$ denotes the canonical basis of $W_\xx^r$.
Given an element $p = \sum_{\aalpha,\bbeta,i} a_{\aalpha,\bbeta,i}\xx^\aalpha\dd^\bbeta e_i$ of $W_{t,\xx}^r$,
we define the \emph{degree} of~$p$~as
\begin{equation*}
\deg(p) = \max\{ |\aalpha| + |\bbeta| \mid \exists i\in\{1,\dots,r\}, \ a_{\aalpha,\bbeta,i}\neq 0\} ,
\end{equation*}
where $|\aalpha|$ and $|\bbeta|$ denote the sums
$\alpha_1+\dots+\alpha_n$ and $\beta_1+\dots+\beta_n$.
Definitions for the algebra~$W_\xx$ mimic the case~$r=1$, simply without considering any~$e_i$,
and definitions for the module~$W_{t,\xx}^r$ are just a special notation when $n$~is replaced with~$n+1$:
a vector basis of~$W_\xx$ is
\begin{equation*}
  \monxx = \{\xx^\aalpha \dd_\xx^\bbeta \mid \aalpha,\bbeta\in \bN^n\} ,
\end{equation*}
and given an element $p = \sum_{\aalpha,\bbeta} a_{\aalpha,\bbeta}\xx^\aalpha\dd_\xx^\bbeta$ of $W_\xx$,
its degree is
\begin{equation*}
\deg(p) = \max\{ |\aalpha| + |\bbeta| \mid a_{\aalpha,\bbeta}\neq 0\} ;
\end{equation*}
a vector basis of~$W_{t,\xx}^r$ is
\begin{equation*}
  \montxxr = \{t^\alpha\xx^\bbeta\partial_t^\gamma\dd_\xx^\ddelta e_i \mid \alpha,\gamma\in \bN, \ \bbeta,\ddelta\in \bN^n, \ i\in\{1,\dots,r\}\} ,
\end{equation*}
and given an element $p = \sum_{\alpha,\bbeta,\gamma,\ddelta,i} a_{\alpha,\bbeta,\gamma,\ddelta,i}t^\alpha\xx^\bbeta\partial_t^\gamma\dd_\xx^\ddelta e_i$ of $W_{t,\xx}^r$,
its degree is
\begin{equation*}
\deg(p) = \max\{ \alpha + |\bbeta| + \gamma + |\ddelta| \mid \exists i\in\{1,\dots,r\}, \ a_{\alpha,\bbeta,\gamma,\ddelta,i}\neq 0\} .
\end{equation*}

\subsection{Holonomic modules}\label{sec:holonomic-modules}

Let $S$ be a submodule of $W_\xx^r$.
We recall the classical definition of a \emph{holonomic $W_\xx$-module}
by means of the \emph{dimension} of the quotient module $M = W_\xx^r/S$.
We point out that any $W_\xx$-module of finite type is isomorphic to a module of this form.
The \emph{Bernstein filtration}~\cite{Bernstein-1971-MRD} of the algebra~$W_\xx$
is the sequence of $\bK$-vector spaces $\cF_m$ defined by
\begin{equation}\label{eq:Bernstein-filtration}
  \cF_m = \left\{ P\in W_\xx \mid \deg(P) \leq m \right\}.
\end{equation}
A filtration of the module~$M$ that is adapted to~$(\cF_m)_{m\geq0}$
is the sequence of $\bK$-linear subspaces $\Phi_m \subseteq M$
defined by
\begin{equation*}
  \Phi_m = \text{image in $M$ of } \left\{ P\cdot e_i \mid P\in \cF_m, \ 1\leq i\leq r \right\}.
\end{equation*}
Those filtrations are compatible with the algebra and module structures:
both $\cF_m \cF_{m'} \subseteq \cF_{m+m'}$ and $\cF_m \Phi_{m'} \subseteq \Phi_{m+m'}$ hold.
There exists a polynomial $p\in \bK[m]$ called the \emph{Hilbert polynomial} of~$M$
that satisfies $\dim_\bK(\Phi_k) = p(k)$ for any sufficiently large $k$.
The \emph{dimension} of the module~$M$ is the degree~$d$
of the polynomial~$p$. The integer~$d$ clearly lies between $0$ and~$2n$.
It was proved by Bernstein that if $M$~is non-zero, then $d$~is larger than or equal to~$n$ \cite[Theorem~9.4.2]{Coutinho_1995}.
When the dimension of~$M$ is exactly~$n$ or when $M$~is the zero module, we say that the module~$M$ is \emph{holonomic}.
(Here, we follow the tradition in \parencite[Chapter~5]{Borel-1987-ADM} and in \parencite{Coutinho_1995}
to consider the zero module as holonomic.
By way of comparison, \parencite{Bjork-1979-RDO}
speaks of a module “in the Bernstein class” to refer to a non-zero holonomic module.)



\subsection{Gröbner bases in Weyl algebras and their modules}\label{sec:grobner-bases-weyl}

Despite their non-commutative nature, by a \emph{monomial} we will mean
an element of the vector bases $\monxxr$, $\monxx$, and~$\montxxr$.
A \emph{monomial order} $\preccurlyeq$ on $W_\xx^r$ is a well-ordering on $\monxxr$
that satisfies for any $i,j\in\{1,\dots,r\}$ and any exponents $\aalpha,\bbeta,\aalpha_1,\bbeta_1,\aalpha_2,\bbeta_2\in\bN^n$
\[
\xx^{\aalpha_1}\dd_\xx^{\bbeta_1} e_i \preccurlyeq \xx^{\aalpha_2}\dd_\xx^{\bbeta_2}e_j \implies \xx^{\aalpha_1 +\aalpha}\dd_\xx^{\bbeta_1+\bbeta}e_i \preccurlyeq \xx^{\aalpha_2+\aalpha}\dd_\xx^{\bbeta_2+\bbeta}e_j.
\]
Given an operator $P\in W_\xx^r$, we define its \emph{support} $\supp(P)$
as the set of all monomials that occur with nonzero coefficient in the decomposition of $P$
with respect to the basis $\monxxr$.
We then define its \emph{leading monomial} $\lm(P)$ as the largest monomial for $\preccurlyeq$ in $\supp(P)$,
its \emph{leading coefficient} $\lc(P)$ as the coefficient of~$\lm(P)$ in this decomposition, and its \emph{leading term} $\lt(P)$ as $\lc(P)\lm(P)$.
We stress that our definition of a leading monomial makes~$\lm(P)$ an element of~$W_\xx^r$, whereas some authors
choose to see leading monomials
as commutative objects in an auxiliary commutative polynomial algebra,
introducing commutative variables~$\xi_i$
to replace the~$\partial_i$ in monomials.
For example, we have~$\lm(\partial_1 x_1 e_i) = x_1 \partial_1 e_i$.
Note that an essentially equivalent theory could be developed by choosing monomials
as elements of the basis consisting of the products~$\dd_\xx^\bbeta \xx^\aalpha e_i$,
instead of monomials in~$\monxxr$.

Computations in Weyl algebras rely heavily on a non-commutative generalization of Gröbner bases.
After the original introduction
\cite{Castro-1984-TDO,Castro-1987-CEI,Galligo-1985-SAQ}
there have been a number of presentations of such a theory, including
\cite{Takayama-1989-GBP,KandriRodyWeispfenning-1990-NGB,LevandovskyySchonemann-2003-PCA}.
A first textbook presentation is~\cite[Chapter~1]{SaitoSturmfelsTakayama-2000-GDH}.
A recent and simpler introduction can be found in~\cite{Bahloul-2020-GBD}.
We now adapt this to our needs.
A Gröbner basis of a left (resp.~right) ideal $I$ of $W_\xx$ with respect to a monomial order $\preccurlyeq$
is a finite set~$G$ of generators of~$I$
such that for any $a\in I$ there exist $g\in G$ and $q\in W_\xx$ satisfying $\lm(a) = \lm(qg)$ (resp.~$\lm(a) = \lm(gq)$).
Note that the non-commutativity of~$W_\xx$ prevents its monomials from enjoying the usual divisibility properties
found in the commutative setting:
in general, the product~$\lm(q) \lm(g)$ is not equal to~$\lm(qg)$, and is not even a monomial;
correlatively, $\lm(a)$~is in general not a multiple of~$\lm(g)$, although $\lm(a) = \lm(qg)$.
However, we always have $\lm(a) = \lm(\lm(q)\lm(g))$.
(The variables~$\xi_i$ introduced by other authors
serve to avoid this formula.)
A Gröbner basis allows to define and compute, for any $a\in W_\xx$, a unique representative of $a+I$ in the quotient $W_\xx/I$
by means of a non-commutative generalization of polynomial division.
We call this unique representative the
\emph{remainder of the division of~$a$ by the Gröbner basis~$G$}
or more shortly
\emph{remainder of~$a$ modulo the Gröbner basis~$G$}.
We denote this remainder $\LRem(a,G)$ when $I$ is a left ideal and $\RRem(a,G)$ when $I$ is a right ideal.
The concept of Gröbner bases for ideals of $W_\xx$ extends to submodules of $W_\xx^r$
in the same way as the notion of Gröbner bases for polynomial ideals generalizes to submodules of a polynomial algebra
(see \cite[Chapter~10.4]{BeckerWeispfenning-1993-GB},
\cite[Chapter~5]{CoxLittleOShea-1998-UAG},
or \cite[Section~3.5]{AdamsLoustaunau-1994-IGB}).
The noetherianity of Weyl algebras implies that any left or right submodule of $W_\xx^r$ admits a (finite) Gröbner basis.

Let $A$ be a subvector space of~$W_\xx^r$.
We define $\dd A$ as the $\bK$-vector space $\sum_{i=1}^n \partial_i A$.
If $A$ is a right $W_\xx$-module, then $\dd A$ is also a  right $W_\xx$-module.

\subsection{Integration}
\label{subsec:intmod}

The \emph{integral of a $W_\xx$-module~$M \simeq W_\xx^r/S$} is the $\bK$-vector space
\begin{equation}\label{eq:intmod}
  M/\dd M \simeq W_\xx^r/(S + \dd W_\xx^r).
\end{equation}
As already mentioned,
it is classical that, if $M$~is holonomic,
then the integral~\eqref{eq:intmod} of~$M$ is a finite-dimensional $\bK$-linear space \cite[Theorem~6.1 of Chapter~1]{Bjork-1979-RDO}.
Computing relations modulo~$\dd M$ is the main matter of this article:
given a family in~$M$, we want to find a linear dependency relation on its image in the integral module~$M/\dd M$,
if any such relation exists.

\subsection{Data structure for holonomic modules}
\label{subsec:data-structure}
Algorithmically, we only deal with holonomic $W_\xx$-modules.
They are finitely presented:
for such a module~$M$,
there exist $W_\xx$-linear homomorphisms~$a$ and~$b$ forming an exact sequence
\[ W_\xx^s \overset{a}{\to} W_\xx^r \overset{b}{\to} M \to 0. \]
Equivalently, this means that~$M \simeq W_\xx^r / S$, where~$S$ is the left submodule generated by the image under~$a$ of the canonical basis of~$W_\xx^s$.
The module~$S$ consists of $W_\xx$-linear combinations of the canonical basis of~$W_\xx^r$,
which we always denote $(e_1,\dots,e_r)$.
This gives a concrete data structure for representing holonomic modules.
It is well known that holonomic modules are \emph{cyclic}, that is, generated by a single element \cite{staffordModuleStructureWeyl1978}.
This means that we could in principle always assume that~$r=1$.
However, some modules have a more natural description with~$r>1$ and transforming the presentation to achieve~$r=1$
has an algorithmic cost that we are not willing to pay.
Therefore, we will not assume~$r=1$.

In order to integrate an infinitely differentiable function $f(\xx)$,
we may consider the $W_\xx$-module generated by~$f$
under the natural action of~$W_\xx$ on $C^\infty$~functions.
Of course, the holonomic approach to symbolic integration will only work if this module is holonomic.
Instead of~$W_\xx\cdot f$, we can also consider any holonomic module that contains it as a submodule.

For example, to integrate a rational function~$A/F \in \bK(\xx)$,
we can consider the module~$\bK[\xx][F^{-1}]$, which is holonomic.
However, finding a finite presentation $W_\xx^r/S \simeq \bK[\xx][F^{-1}]$ is not trivial.
There are algorithms \cite{OakuTakayamaWalther-2000-LAD} to solve this problem, but,
in terms of efficiency, it is still a practical issue
that we do not address in this work.
Fortunately, in many cases
we can easily construct a holonomic module for integration.
See for example Theorem~\ref{thm:kreg-main} in Section~\ref{sec:kreg}.

\section{Reductions}
\label{sec:reduction}

We consider the Weyl algebra~$W_\xx$ over a field $\bK$
and a finitely presented $W_\xx$-module~$M$ given in the form~$W_\xx^r/S$ for some~$r\geq1$ and some submodule~$S$ of~$W_\xx^r$ (see Section~\ref{sec:holonomic-modules}).
The main objective of the section is to compute
normal forms in~$M$ modulo~$\dd M$, or, equivalently, normal forms in~$W_\xx^r$ modulo~$S + \dd W_\xx^r$.
In other words, we want an algorithm that given some~$a\in W_\xx^r$ computes some~$[a]\in W_\xx^r$,
and such that $[a] = [b]$ if and only if~$a - b \in S + \dd W_\xx^{r}$.
This goal is only partially reached with a family of reductions~$[.]_\eta$
such that for each pair~$(a,b)$, there exists~$\eta$
such that $[a]_\eta = [b]_\eta$ if and only if~$a-b \in S + \dd W_\xx^r$.
The existence of a monomial~$\eta$ is not effective,
similarly to the maximal total degree to be considered in Takayama's algorithm~\cite{Takayama-1990-ACI}.
This is a step backwards compared to previous methods \cite{OakuTakayama_2001},
but computing weaker normal forms allows for more efficient computational methods.
Concretely, we do not rely on the computation of $b$-functions.

The present section is organized as follows.
In Section~\ref{subsec:red} we define a reduction procedure $[.]$
that partially reduces elements of~$W_\xx^r$ by $S + \dd W_\xx^r$,
in the sense that the procedure will in general not reduce every element of $S + \dd W_\xx^r$ to zero.
In Section~\ref{subsec:compirred} we define a filtration $(F_{\preccurlyeq\eta})_{\eta\in\monxxr}$ of the vector space~$S + \dd W_\xx^r$
and we give an algorithm to compute a basis of each vector space~$[F_{\preccurlyeq\eta}]$ of reduced forms.
Using this basis we define a new reduction~$[.]_\eta$ that enhances the first one.
In Section~\ref{subsec:guess_eta},
we provide, for some infinite families~$(a_i)_{i\geq 0}$ in~$W_\xx^r$, an algorithm for computing an~$\eta$ such that all the~$[a_i]_\eta$ lie in a finite-dimensional subspace.
In Section~\ref{subsec:GD}, we consider the case
where $S$~is the annihilator of~$e^f$ for some homogeneous multivariate polynomial~$f$,
and we compare our reduction procedures with variants of the Griffiths--Dwork reduction.

\subsection{Reduction $[.]$ and irreducible elements}
\label{subsec:red}

\paragraph{Reduction rules.}

Let $\preccurlyeq$ be a monomial order on $W_\xx^r$ and let $G$ be a Gröbner basis of $S$ for this order.
We define two binary relations $\to_1$ and $\to_2$ on $W_\xx^r\times W_\xx^r$ as follows:
\begin{itemize}
\item
  Given $a \in W_\xx^r$, $\lambda\in\bK$, $m\in \monxx$, and $g\in G$,
  we write
  \[ a \to_1 a -  \lambda m g \]
  if $\lm(mg)$ is in the support of~$a$ but not in the support of $a- \lambda m g$.
\item
  Given $a \in W_\xx^r$, $\lambda\in\bK$, $m\in \monxxr$, and $i\in\{1,\dots,n\}$,
  we write
  \[ a \to_2 a - \lambda \partial_i m \]
  if $\lm(\partial_i m)$~is in the support of~$a$ but not in the support of $a-\lambda\partial_i m$.
\end{itemize}

The relation $\to_1$ corresponds
to the reduction by the Gröbner basis~$G$
of the left module~$S$
and the relation $\to_2$ corresponds
to the reduction by the Gröbner basis $\{\partial_i e_j \mid i=1,\dots,n, \ j=1,\dots,r\}$
of the right module~$\dd W_\xx^r$.
Next, we define~$\to$ as the relation $\to_1\cup\to_2$.
That is, $a \to b$ if either~$a\to_1 b$ or~$a\to_2 b$.
The relation~$\to^+$ is the transitive closure of~$\to$:
$a\to^+ b$ if there exist an integer $s\geq1$ and a sequence of $s$~reductions
\begin{equation}\label{eq:a-red-b}
a \to c_1 \to \dots \to c_{s} = b
\end{equation}
for some~$c_1,\dotsc,c_{s} \in W_\xx^r$.
The relation~$\to^*$ is the reflexive closure of~$\to^+$: $a\to^* b$ if either~$a \to^+b$ or~$a = b$.
In this situation we say that \emph{$a$~reduces to~$b$}.

\paragraph{Irreducible elements.}

We say that an element~$b$ is \emph{irreducible} if there is no~$c$ such that $b \to c$
and we say that $b$ is a \emph{reduced form of~$a$} if $b$~is irreducible and~$a \to^* b$.

\begin{lemma}\label{lem:redinlrm}
Let $a,b \in W_\xx^r$.
If $a \to^* b$ then $a-b\in S + \dd W_\xx^r$.
\end{lemma}

\begin{proof}
This follows from the definition of~$\to_1$ and~$\to_2$ since the terms~$mg$ and~$\partial_i m$ are in~$S$ and~$\dd W_\xx^r$, respectively.
\end{proof}

However, the converse of~Lemma~\ref{lem:redinlrm} is not true in general, even when~$b = 0$: there may be nonzero irreducible elements in~$S+\dd W_\xx^r$.

\begin{lemma}\label{lem:vp}
The irreducible elements of~$W_\xx^r$ form a $\bK$-vector space.
\end{lemma}

\begin{proof}
The set~$V$ of all irreducible elements contains~$0$ and is stable by multiplication by $\bK$.
Let $a,b\in V$ and assume by contradiction that $a+b$ is not irreducible.
Then, there exists a monomial~$m \in \monxxr$ in the support of~$a+b$ that can be reduced by~$\to$.
Because $a+b=b+a$, we can without loss of generality
assume that $m$~is in the support of $a$.
This contradicts the irreducibility of~$a$.
Thus $a+b\in V$.
\end{proof}

The vector space of Lemma~\ref{lem:vp} can be infinite-dimensional,
as we now exemplify.

\begin{example}\label{ex:badredcase}
  Let $S = W_{x_1}\partial_1$ be the left ideal generated by $\partial_1$ in the Weyl algebra in one pair of generators,~$(x_1,\partial_1)$.
  Note that~$W_{x_1}/S \simeq \bK[x_1]$ as~$W_{x_1}$-module.
  Let $\preccurlyeq$ be the lexicographic order for which $\partial_1 \preccurlyeq x_1$.
  Then, any element of~$\bK[x_1]$ (as a subspace of~$W_{x_1}$) is irreducible.
\end{example}

Irreducible forms can be computed by alternating left reductions with respect to a Gröbner basis of~$S$ (representing the rule~$\to_1$)
and right reductions with respect to a Gröbner basis of~$\dd W_\xx^r$ (representing the rule~$\to_2$).
This leads to Algorithm~\ref{algo:irred-form}.
Correctness is clear.
The algorithm terminates since the largest reducible monomial in~$a$, if any, decreases at each iteration of the loop.

\begin{alg}{Computation of a reduced form}{irred-form}
  \textbf{Input:}
  \begin{algovals}
    \item $a\in W_\xx^r$
    \item a Gröbner basis $G$ of $S$
  \end{algovals}
  \textbf{Output:}
  \begin{algovals}
    \item a reduced form of $a$
  \end{algovals}
  \begin{pseudo}
    \kw{while} $a$ is not irreducible \\+
     $a \gets \RRem(a, \{\partial_i e_j \mid i=1,\dots,n, \ j=1,\dots,r\})$ \\
      $a \gets \LRem(a, G)$ \\-
    \kw{return} $a$ \\
  \end{pseudo}
\end{alg}

\begin{definition}
We denote by~$[a]$
the reduced form of~$a\in W_\xx^r$ that is computed by Algorithm~\ref{algo:irred-form}.
\end{definition}

\begin{proposition}\label{prop:[.]-is-linear}
The map~$[.]$ is $\bK$-linear.
\end{proposition}

\begin{proof}
The maps RRem and LRem are $\bK$-linear by the uniqueness of the remainder
of a division by a Gröbner basis.
Let $\tau(a)$ denote
the number of iterations of the while loop in Algorithm~\ref{algo:irred-form} on input~$a$.
Given $\tau\in\bN$, let $V_\tau$ denote the set of all~$a$ for which $\tau(a) \leq \tau$.
The restriction of~$[.]$ on~$V_\tau$ takes the same values as
the composition of $\tau$~copies of RRem and $\tau$~copies of LRem in alternation,
in which some of the final copies effectively act by the identity as they input irreducible elements.
So the restriction of~$[.]$ on~$V_\tau$ is $\bK$-linear as a composition of linear maps.
The result follows because $W_\xx^r = \bigcup_{\tau\geq0} V_\tau$.
\end{proof}

\begin{example}\label{ex:airy-red}
The following identity will serve as a running example throughout the article:
\begin{equation*}
\int_{\gamma}
    \exp\left(\frac{x^{3}}{3} - x(t+2z)\right)
    \exp\left(\frac{y^{3}}{3} - y(t+z)\right)
    \, dx\, dy\, dz
    = \frac{(2\pi i)^2}{7^{1/3}\sqrt{\pi}}
      \ai\left(\frac{t}{7^{1/3}}\right),
\end{equation*}
where the integration domain $\gamma$ is the product
of $\bR$ with two loops around~$0$ in~$\bC$.
It can be obtained by classical means, e.g.,
it is obtained by specializing the parameters $(a,b,c,d)$ to $(2,t,1,t)$ in~\cite[Eq.~2.11.6]{prudnikovIntegralsSeries31986},
leading to
\begin{equation*}
\int_{-\infty}^{\infty} \ai(2z+t) \ai(z+t) \, dz = \frac{1}{7^{1/3}\sqrt{\pi}} \ai\left(\frac{t}{7^{1/3}}\right) ,
\end{equation*}
before replacing twice the Airy function with its contour integral representation
\begin{equation*}
\ai(u) = \frac{1}{2\pi i} \oint \exp\left(\frac{x^{3}}{3} - xu\right)\, dx .
\end{equation*}

The triple-integral identity can also be obtained with the algorithms of this paper.
Let
\[ f(x,y,z,t)
    = \exp\left(\frac{x^{3}}{3} - x(t+2z)\right)
      \exp\left(\frac{y^{3}}{3} - y(t+z)\right). \]
Write $W_{x,y,z}(t)$ for the Weyl algebra over the base field $\bK = \bQ(t)$, with~$n=3$.
Consider the operators
$\partial_x - f^{-1}\tfrac{\partial f}{\partial x}$,
$\partial_y - f^{-1}\tfrac{\partial f}{\partial y}$,
$\partial_z - f^{-1}\tfrac{\partial f}{\partial z}$,
which are elements of~$W_{x,y,z}(t)$.
They generate the annihilating ideal~$\ann(f)$ of~$f$ in~$W_{x,y,z}(t)$.
We then form the module $M = W_{x,y,z}(t) / \ann(f)$.
Note that we do not consider here the differential structure with respect to~$t$.
We fix a monomial order that is a block order in the sense of \parencite{BeckerWeispfenning-1993-GB}:
monomials are compared first according to the grevlex (graded reverse lexicographic) order satisfying $x>y>z$;
ties are next broken by comparing according to the grevlex order satisfying $\partial_x>\partial_y>\partial_z$.
The reduced Gröbner basis $G$ of $\ann(f)$ with respect to this monomial order%
\footnote{This is $\grevlex(x,y,z) \;>\; \grevlex(\partial_x,\partial_y,\partial_z)$ in the notation of
Brochet's Julia implementation.
This corresponds to \texttt{lexdeg([x,y,z],[$\partial_x$,$\partial_y$,$\partial_z$])} in Maple's notation.}
is given by:
\begin{align*}
g_1 =&  \underline{y^{2}} - z - \partial_y - t, \\
g_2 =& 14\underline{yz} + 8y\partial_x - 2y\partial_y + 6ty
 - 11 z\partial_z
 + \partial_z^{3}
 - 4\partial_x\partial_z
 - 3\partial_y\partial_z
 - 7t\partial_z
 - 11, \\
g_3 =& 49 \underline{z^{2}} + 14y
 - 18 z\partial_z^{2}
 + 56 z\partial_x
 - 14 z\partial_y
 + 42t z
 + \partial_z^{4}
 - 8\partial_x\partial_z^{2}, \\
 &- 2\partial_y\partial_z^{2}
 + 16\partial_x^{2}
 - 8\partial_x\partial_y
 + \partial_y^{2}
 - 10t\partial_z^{2}
 + 24t\partial_x
 - 6t\partial_y
 - 20\partial_z
 + 9t^{2} \\[4pt]
g_4 =& 2\underline{x} + y + \partial_z, \\
g_5 =& 2\underline{y\partial_z}
 - 7 z
 + \partial_z^{2}
 - 4\partial_x
 + \partial_y
 - 3t.
\end{align*}

Let us now compute the reduction $[y^2]$ (Algorithm~\ref{algo:irred-form}). The monomial $y^2$ is not reducible by $\dd W_{x,y,z}(t)$ as it does not involve any $\partial$.
Then we reduce $y^2$ modulo $\ann(f)$:
\begin{equation}
  \LRem(y^2,G) = y^2 - g_1 = z + \partial_y + t.
\end{equation}
Then, we reduce the result modulo $\dd W_\xx$:
\begin{equation}
  \RRem(z + \partial_y  + t, \{\partial_x,\partial_y,\partial_z\}) = (z + \partial_y  + t) - \partial_y = z + t.
\end{equation}
The operator $z + t$ is now irreducible, hence $[y^2] = z + t$.
\end{example}

\subsection{Computation of the irreducible elements of $S + \dd W_\xx^r$}
\label{subsec:compirred}

Again, we fix an order~$\preccurlyeq$
and a submodule~$S$ of~$W_\xx^r$ by considering a Gröbner basis~$G$ of it.
Let $\Irr$ be the vector space of all irreducible elements of $S + \dd W_\xx^r$. This vector space
can be infinite-dimensional hence we cannot hope to compute all of it.
We therefore define a vector-space filtration $(F_{\preccurlyeq\eta})_{\eta\in\monxxr}$
of $S + \dd W_\xx^r$~by
\[
  F_{\preccurlyeq\eta} = \left\{s + d \in W_\xx^r \mid s \in S, \ d\in \dd W_\xx^r \text{, and } \max (\lm(s),\lm(d)) \preccurlyeq \eta \right\} ,
\]
and a vector-space filtration of~$\Irr$ by $\Irr_{\preccurlyeq\eta} := F_{\preccurlyeq\eta}\cap\Irr$.
We define $F_{\prec\eta}$ and~$\Irr_{\prec\eta}$ similarly,
by requiring a strict inequality on the maximum of the leading monomials.
%

Our goal is to obtain an efficient computation of a $\bK$-basis of $\Irr_{\preccurlyeq\eta}$.
Let us give an intuitive description of our algorithm.
By general properties of Gröbner bases, a non-zero element reduces to zero using the relation~$\rightarrow_1$
(resp.~$\rightarrow_2$) if and only if it belongs to~$S$ (resp.~$\dd W_\xx^r$).
The difficulty arises when both reduction rules can be applied to reduce a monomial.
For example, take an element~$s$ in~$S$ such that $\lm(s) \in \lm(S)\cap \lm(\dd W_\xx^r)$
and, assuming it can be reduced so as to cancel its leading monomial by using $\rightarrow_2$,
perform this reduction, that is,
find~$s'$ such that~$s \rightarrow_2 s'$ and $\lm(s') \prec \lm(s)$.
In this case, it is possible that $s'$~is neither in~$S$ nor in~$\dd W_\xx^r$,
making it a good candidate for an element that does not reduce to~$0$ by~$\rightarrow$.
The following theorem shows more precisely how generators of~$\Irr$ can be obtained.

\begin{theorem}\label{thm:redT}
Let $\eta \in \monxxr$.
\begin{enumerate}
\item If $\eta \not\in \lm(S)\cap \lm(\dd W_\xx^r)$, then
$
\Irr_{\preccurlyeq \eta}= \Irr_{\prec\eta}$.

  \item If $\eta\in \lm(S)\cap \lm(\dd W_\xx^r)$, then
        $
        \Irr_{\preccurlyeq \eta} =  \Irr_{\prec\eta} + \bK a$,
        for any reduced form~$a$ of~$mg - \partial_i w$,
        where~$w \in W_\xx^r$, $m\in W_\xx$ and $g\in G$ are any elements such that $\eta = \lm(mg) = \lm(\partial_i w)$ and\/~$\lc(mg) = \lc(\partial_i w)$.
Moreover, such $m$ and~$g$ exist because~$G$ is a Gröbner basis of~$S$.
\end{enumerate}
\end{theorem}

\begin{proof}
For the first point, we prove by contradiction that
for any $a\in \Irr_{\preccurlyeq \eta}$ and any $s\in S$ and $d\in \dd W_\xx^r$
satisfying $a = s+d$ and $\max(\lm(s),\lm(d)) \preccurlyeq \eta$,
we have in fact $\max(\lm(s),\lm(d)) \prec \eta$.
This will imply the equality $\Irr_{\preccurlyeq \eta}= \Irr_{\prec\eta}$.
Let us assume that the equality $\max(\lm(s),\lm(d)) = \eta$ holds.
Therefore, either $\lm(s) = \lm(d) = \eta$, or $\lm(s)\prec \lm(d) = \eta$, or $\lm(d)\prec\lm(s) = \eta$.
The first case is excluded because we assumed $\eta \not\in \lm(S)\cap \lm(\dd W_\xx^r)$.
In both remaining cases it is possible to reduce $\lm(a)=\eta$ with one of the two reduction rules.
This contradicts the fact that $a$~is irreducible.

For the second point, let $m$, $g$, $w$, $i$, and~$a$ be given as in the statement.
We first check that~$\Irr_{\prec\eta} + \bK a \subseteq \Irr_{\preccurlyeq\eta}$.
It is enough to prove that~$a \in \Irr_{\preccurlyeq\eta}$.
By definition, $mg - \partial_i w \in F_{\preccurlyeq \eta}$,
and we check easily that~$F_{\preccurlyeq \eta}$ is stable under~$\to$.
So~$a\in F_{\preccurlyeq}$. Since~$a$ is also irreducible, we have~$a\in \Irr_{\preccurlyeq\eta}$.

Let us prove the other inclusion.
Let $f\in \Irr_{\preccurlyeq\eta}$.
Then $f$~is irreducible and of the form~$s + d$ for $s\in S$ and $d\in \dd W_\xx^r$ satisfying $\max(\lm(s),\lm(d)) \preccurlyeq \eta$.
If this inequality is strict, then $f\in \Irr_{\prec \eta}$,
proving $f \in \Irr_{\prec\eta} + \bK a$.
Otherwise, we have the equality~$\max(\lm(s), \lm(d)) = \eta$.
Let us remark the equality $\lm(s) = \lm(d)$, for otherwise
either $\lm(s) \succ \lm(d)$ and $f$~could be reduced using~$\rightarrow_1$,
or $\lm(s) \prec \lm(d)$ and $f$~could be reduced using~$\rightarrow_2$.
So $\eta = \lm(s) = \lm(d)$.
This monomial cannot be~$\lm(f)$,
for otherwise $f$~could be reduced using any of $\rightarrow_1$ and~$\rightarrow_2$.
Hence $\lm(f) \prec \eta$ and $\lt(s) = - \lt(d)$.
We decompose $s$ and~$d$ as
$s = \lambda mg + s'$ and $d =  -\lambda \partial_i w + d'$ with $\lambda\in \bK$, $s'\in S$, $d'\in \dd W_\xx^r$, and $\max(\lm(s'), \lm(d')) \prec \eta$.
Let $h$ denote~$mg - \partial_i w$,
which, by hypothesis, has~$a$ as a reduced form.
This implies an equality of the form $h = a + s'' + d''$
with $s'' \in S$, $d''\in \dd W_\xx^r$, and $\max(\lm(s''), \lm(d'')) \prec \eta$.
We obtain $f = s + d = \lambda (mg - \partial_i w) + s' + d' = \lambda a + b$ with $b = s' + \lambda s'' + d' + \lambda d''$.
Since both $f$ and $a$ are irreducible so is $b$, thus $b\in \Irr_{\prec\eta}$,
proving that $f$~is in $\Irr_{\prec\eta} + \bK a$.
\end{proof}

The meaning of Theorem~\ref{thm:redT} is that
the dimension of the filtration $(\Irr_{\preccurlyeq\eta})_\eta$ is susceptible
to increase at~$\eta$ only if $\eta\in \lm(S)\cap \lm(\dd W_\xx^r)$.
But this is not necessary as the element~$a$ may well be in~$\Irr_{\prec\eta}$.
The following lemma describes a sufficient condition for this situation.

\begin{lemma}\label{lem:crit-red}
  Let $\eta\in \lm(S)\cap\lm(\dd W_\xx^r)$. If there exist $g\in G$, $m\in\monxx$, and some~$i$
  such that $\eta = \lm(\partial_i mg)$, then
  \[
  \Irr_{\prec \eta} = \Irr_{\preccurlyeq \eta}.
  \]
\end{lemma}

\begin{proof}
  By Theorem~\ref{thm:redT}, the result reduces to proving that
  $\Irr_{\prec \eta}$~contains the reduced form~$a$ of some $h = \partial_i mg - \partial_j w$ with~$\lm(h) \prec \eta$.
  We choose~$j = i$ and~$w = mg$, so that~$h=0$, which already is irreducible and in~$\Irr_{\prec \eta}$.
\end{proof}



\begin{corollary}\label{cor:Irr-eta}
  Let $\eta \in \monxxr$.
  Let $H$ be the set of monomials~$m \preccurlyeq \eta$ such that $m\in \lm(S)\cap\lm(\dd W_\xx^r)$ and
  $m \neq \lm(\partial_i p g)$
  for any $i\in \{1,\dots, n\}$, $g\in G$, and $p\in\monxx$.
  For~$m\in H$, let $a_m \in W_\xx^r$ be any reduced form of some $\xx^\ggamma g - \lc(g) \dd^\bbeta \xx^{\aalpha + \ggamma} e_j$, where
  $g \in G$,
  $\lm(g) = \xx^{\aalpha} \dd^{\bbeta} e_j$,
  and $m = \lm(\xx^\ggamma g)$.
  Then
  \begin{equation}\label{eq:Irr-as-sum}
    \Irr_{\preccurlyeq \eta} = \sum_{m\in H} \bK a_{m} .
  \end{equation}
\end{corollary}

\begin{proof}
  Note that for each~$m\in H$, the corresponding~$\bbeta$ is nonzero.
  Indeed, by definition, $m\in \lm(\dd W_x^r)$ so there is some~$\partial_i$ such that~$m = \lm(\partial_i m')$ for another monomial~$m'$. Moreover, $m = \lm(\xx^\ggamma g)$, so~$\lm(g)$ also contains~$\partial_i$.
  In particular, the term $\lc(g) \dd^\bbeta \xx^{\aalpha + \ggamma} e_j$ has the form~$\partial_i w$.
  Therefore,
  Theorem~\ref{thm:redT} applies and
  $\Irr_{\preccurlyeq m} =  \Irr_{\prec m} + \bK a_m$.
  For a monomial~$m$ not in~$H$, either Theorem~\ref{thm:redT} or Lemma~\ref{lem:crit-red}
  shows that~$\Irr_{\preccurlyeq m} =  \Irr_{\prec m}$.
  Then the statement follows by well-founded induction on~$\eta$.
\end{proof}

To turn Corollary~\ref{cor:Irr-eta} into an algorithm,
we introduce a finiteness property of the monomial order~$\preccurlyeq$.
\begin{hypothesis}\label{hyp:finiteness-for-eta}
For any two monomials $\gamma$ and~$\eta$ of~$\monxxr$,
the set of~$\aalpha$ for which $\xx^\aalpha \gamma \preccurlyeq \eta$ is finite.
\end{hypothesis}
This hypothesis is always satisfied by orders graded by total degree,
because a monomial~$\eta$ has a finite number of predecessors in~$\monxxr$.
It is also satisfied by orders eliminating~$\xx$, in the sense that
\begin{equation}\label{eq:elimination-order}
\aalpha'-\aalpha \in \bN^n \setminus \{0\} \Rightarrow \xx^\aalpha \dd_\xx^\bbeta e_i \prec \xx^{\aalpha'} \dd_\xx^{\bbeta'} e_{i'} ,
\end{equation}
as long as the set of~$\aalpha$ for which $\xx^\aalpha \preccurlyeq x_i$ is finite for each $i \in \{1,\dots,n\}$.
For example, this contains “elimination orders” \parencite{CoxLittleOShea-1998-UAG} or “block orders” \parencite{BeckerWeispfenning-1993-GB} that first order by total degree in~$\xx$,
but not a lexicographical order that has $x_1 > x_2 > \partial_1 > \partial_2$.

\begin{alg}{Computation of $\Irr_{\preccurlyeq\eta}$}{basis-E-lambda}
  \textbf{Input:}
  \begin{algovals}
    \item a Gröbner basis~$G$ of~$S$
    \item $\eta\in \monxxr$
    \end{algovals}
  \textbf{Output:}
  \begin{algovals}
    \item a generating family of the $\bK$-vector space $\Irr_{\preccurlyeq\eta}$
  \end{algovals}
  \begin{pseudo}
    $G' \gets \left\{ g \in G \st \lm(g) \text{ involves some $\partial_i$} \right\}$\\
    \label{line:defH}$H \gets \left\{ \lm(\xx^\ggamma g) \st  \ggamma \in \mathbb{N}^n, \ g\in G', \text{ and } \lm(\xx^\ggamma g)\preccurlyeq \eta \right\}$ \ct{finite by Hypothesis~\ref{hyp:finiteness-for-eta}}\\
    $H \gets H \setminus \left\{ \lm(\partial_i m g) \st m\in \monxx, \ g\in G, \ 1\leq i\leq n \right\}$\\

    $B \gets \varnothing$ \\
    \kw{for} $m \in H$\\+
      pick $g\in G'$ and $\ggamma$ such that~$m = \lm(\xx^\ggamma g)$\\
      $\xx^{\aalpha} \dd^\bbeta e_i \gets \lm(g)$ \quad \ct{by construction~$\bbeta \neq 0$}\\
      \label{line:reduce-B}$B \gets B \cup \big\{ \big[\xx^\ggamma g - \lc(g) \dd^\bbeta \xx^{\aalpha + \ggamma} e_i \big] \big\}$\\-
    \kw{return} $B$ \\
  \end{pseudo}
\end{alg}

\begin{theorem}\label{thm:algo-basis-E-lambda-is-correct}
Under Hypothesis~\ref{hyp:finiteness-for-eta}, Algorithm~\ref{algo:basis-E-lambda} is correct and terminates.
\end{theorem}

\begin{proof}
  Termination is obvious since the set on line~\ref{line:defH} is finite, by hypothesis.
  For the correctness, we observe that~$H$ computed in the algorithm is the same as the set~$H$ described in Corollary~\ref{cor:Irr-eta}.
\end{proof}

\begin{definition}
Let $B_\eta$ be an echelon form of the generating family returned by Algorithm~\ref{algo:basis-E-lambda} on input $\eta$.
We define a reduction $[.]_{\eta}$ from $W_\xx^r$ into itself by
\[
[a]_\eta = \Reduce( [a], B_\eta)
\]
where $[.]$ is the map defined by Algorithm~\ref{algo:irred-form}
and $\Reduce(.,B_\eta)$ is the reduction algorithm by the echelon form $B_\eta$.
\end{definition}

\begin{proposition}
The map~$[.]_\eta$ is $\bK$-linear.
\end{proposition}

\begin{proof}
This follows from Proposition~\ref{prop:[.]-is-linear} and the $\bK$-linearity of $\Reduce(.,B_\eta)$.
\end{proof}

\begin{theorem}
\label{thm:red}
For any $a\in S + \dd W_\xx^r$ there exists $\eta\in \monxxr$ such that for all $\eta' \succcurlyeq \eta$, the remainder~$[a]_{\eta'}$ is zero.
\end{theorem}

\begin{proof}
  The element~$[a]$ is congruent to~$a$ modulo~$S+\dd W_\xx^r$,
  so it is in~${S+\dd W_\xx^r}$, like~$a$ itself.
  Moreover, it is irreducible, and so it is in~$\Irr$ by the definition of~$\Irr$.
  Because of the equality $\Irr = \bigcup_{\eta \in \monxxr} \Irr_{\preccurlyeq\eta}$, there exists $\eta$
  such that $[a]\in \Irr_{\preccurlyeq\eta}$
  and thus $\Reduce([a], B_\eta) = 0$.
  For~$\eta' \succcurlyeq \eta$, the vector space $\Span_\bK(B_\eta)$ is included in~$\Span_\bK(B_{\eta'})$, so
  $[a]_{\eta'} = \Reduce([a], B_{\eta'}) = \Reduce(\Reduce([a], B_\eta), B_{\eta'}) = \Reduce(0, B_{\eta'}) = 0$.
\end{proof}

\begin{definition}\label{def:nf}
The \emph{normal form} of an element~$a \in W_\xx^r$ modulo~$S+\dd W_\xx^r$ is the unique element~$a' \in W_\xx^r$
such that~$a \equiv a' \pmod{S+\dd W_\xx^r}$
and no monomial of~$a'$ is the leading monomial of an element of~$S+\dd W_\xx^r$.
\end{definition}

\begin{corollary}\label{cor:nf}
  For any~$a\in W_\xx^r$, there exists $\eta\in \monxxr$ such that for all $\eta' \succcurlyeq \eta$,
  the remainder~$[a]_{\eta'}$ is the normal form of~$a$ modulo~$S+\dd W_\xx^r$.
\end{corollary}

\begin{proof}
  Let~$a'$ be the normal form of~$a$.
  Let~$\eta$ such that~$[a- a']_\eta = 0$, given by Theorem~\ref{thm:red}.
  By definition, $[.]_\eta$~replaces monomials by smaller ones,
  but only if this is possible,
  so that we have~$[a']_\eta = a'$.
  By linearity of~$[.]_\eta$, we obtain~$[a]_\eta = a'$.
\end{proof}

\begin{example}\label{ex:airy-irred}
Let us continue Example~\ref{ex:airy-red}.
Fix $\eta = x^{2}$, the largest monomial of degree~$2$ for our monomial order.
We wish to compute $\Irr_{\preccurlyeq \eta}$ by Algorithm~\ref{algo:basis-E-lambda}.
The set~$G'$ is first set to~$\{g_5\}$, and after line~\ref{line:defH}, the set~$H$ equals $\{y\partial_{z}\}$:
this single element is obtained for~$g=g_5$, the only possibility for~$g$,
and $\xx^\ggamma=1$, as any other choice of~$\xx^\ggamma$ would imply
$\lm(\xx^\ggamma g)\succ \eta$.
Then, line~\ref{line:reduce-B} computes $[g_5 - 2\partial_z y]$, which evaluates to $-(7z + 3t)$.
We therefore deduce by Theorem~\ref{thm:algo-basis-E-lambda-is-correct} that
\[
  \Irr_{\preccurlyeq x^2} = \operatorname{Span}_{\bQ(t)} \left\{ 7z + 3t \right\}.
\]

In Example~\ref{ex:airy-red}, we had computed~$[y^2]$ and obtained~$z + t$.
We can now push this reduction further by computing the $\eta$-reduction of~$y^2$ as follows:
\begin{equation}
  [y^{2}]_{x^2}
  = \Reduce([y^{2}], \{7z + 3t\})
  = (z + t) - \tfrac{1}{7}(7z + 3t)
  = \tfrac{4}{7}t.
\end{equation}

\end{example}

\paragraph{Dimension of the vector space $\Irr_{\preccurlyeq\eta}$.}

By the finite dimensionality of $M/\dd M$, only finitely many monomials of $\monxxr$ are irreducible modulo $S + \dd W_\xx^r$.
As a consequence,
every monomial~$m \in \monxxr$ except for the finitely many irreducible ones
is either reducible by~$G$ or by an echelon form of~$\Irr_{\preccurlyeq \eta}$ for some~$\eta$.
If an infinite number of monomials is not reducible by~$G$, as in Example~\ref{ex:badredcase}, the dimension
of~$\Irr_{\preccurlyeq\eta}$ will tend to infinity when $\eta$~increases indefinitely,
making the computation of~$\Irr_{\preccurlyeq\eta}$ increasingly expensive.
As a consequence, the computational cost of Algorithm~\ref{algo:basis-E-lambda}
depends on the structure of the staircase formed by the leading monomials of~$G$.
We present  two extreme scenarios: in one, $\Irr_{\preccurlyeq\eta}$ is equal to $\{0\}$ for any $\eta$
and in the other, no monomial of $\bK[\xx]$ is reducible by~$G$.
Naturally, intermediate cases also exist.

\begin{example}
Let $S$ be the left ideal of $W_\xx(t)$ generated by the Gröbner basis
\begin{equation*}
  (t - 1)\underline{x_1} - t\partial_1 , \qquad
  \underline{x_2} -t .
\end{equation*}
Up to renaming variables, this is the left ideal used for the computation of the generating series of $2$-regular graphs in Section~\ref{sec:kreg}.
Every operator in this Gröbner basis has its leading monomial in $\bK[\xx]$, therefore  Theorem~\ref{thm:redT}
implies that $\Irr_{\preccurlyeq\eta}$ is~$\{0\}$ for any~$\eta$.
In this very special case we obtain that the reduction~$[.]$ computes normal forms. That is, $[a] = 0$
if and only if $a\in S + \dd W_\xx^r$.
We observed the same phenomenon with the ideals $S$ defined in Theorem~\ref{thm:kreg-main} for $k$-regular graphs up to $k=8$.
\end{example}

\begin{example}
Let $f(\xx,t) = 1 - (1-x_1x_2)x_3 - tx_1x_2x_3(1-x_1)(1-x_2)(1-x_3)$.
The integral of $1/f$ is related to the generating function of Apéry numbers~\cite{Beukers-1983-IPP}.
We were able to compute, by a method that we do not describe here,
a Gröbner basis for the grevlex order of a $W_\xx(t)$-ideal~$S$ included in $\ann(1/f)$ such that $W_\xx(t)/S$ is holonomic.
This Gröbner basis contains 26 operators but all of their leading monomials contain a~$\partial_i$.
Hence, no monomial of $\bK[\xx]$ is reducible by the Gröbner basis of $S$.
We observed the same phenomenon for every rational function that we tried.
\end{example}

Lastly, we present the simplest example on which our reduction $[.]_\eta$ is inefficient.

\begin{example}\label{ex:badredcase-continued}
We continue Example~\ref{ex:badredcase},
in which we set $S$ to~$W_{x_1}\partial_1$
and $\preccurlyeq$ to the lexicographic order $\partial_1 \preccurlyeq x_1$.
We remarked that  $\Irr \subset \bK[x_1]$, and reciprocally the equality
$(i+1)x_1^i = \partial_1 x_1^{i+1} - x_1^{i+1}\partial_1$ proves that $\bK[x_1] \subset \Irr$.
The reduction $[.]$ does not see that elements of $\bK[x_1]$ are reducible by $S + \dd W_{x_1}$
and Algorithm~\ref{algo:basis-E-lambda} ends up calculating a basis of $\bK[x_1]_{\preccurlyeq\eta}$.
\end{example}

\subsection{Confinement}
\label{subsec:guess_eta}

``Computing'' in the quotient $M/\dd M \simeq W_\xx^r/(S + \dd W_\xx^r)$
can take on several forms, with various levels of potency.
In the strongest interpretation, we want to compute a basis of the quotient, as a $\bK$-linear space,
and we want to be able to compute normal forms in~$W_\xx^r$ modulo~$S + \dd W_\xx^r$.
In a weaker sense, we merely want to be able to capture the finiteness of the quotient space,
without ensuring the linear independence of a finite generating set or even producing it explicitly.
In view of our needs for integration algorithms in the next sections,
there is an even weaker sense:
given $a \in W_\xx^r$ (which will designate a function to be integrated)
and a $W_\xx$-linear map~$L$ from~$W_\xx^r$ to itself (which will be related to taking derivatives with respect to a parameter $t \in \bK$),
we need to testify the finite-dimensionality of the span over~$\bK$
of the orbit $\{L^i(a)\mid i\in\bN\}$ modulo $S+\dd W_\xx^r$.
In practice (and in particular in Algorithm~\ref{algo:confinement}),
the map~$L$ will be provided by a square matrix~$\Lambda$ with entries in~$W_\xx$,
such that $L(a) = a \Lambda$.
In this section, we show that the reductions~$[.]_\eta$ can be used to find,
for any $a$ and~$L$, a finite set~$B$ that witnesses this finite-dimensionality.

\begin{alg}{Computation of a confinement}{confinement}
  \textbf{Input:}
  \begin{algovals}
    \item a module $S\subseteq W_\xx^r$ given by a Gröbner basis
    \item $a \in W_\xx^r$
    \item a $W_\xx$-linear map $L\colon W_\xx^r\to W_\xx^r$ given by an $r\times r$ matrix~$\Lambda$
    \item $\rho \in \bN$
    \end{algovals}
  \textbf{Output:}
  \begin{algovals}
    \item an effective confinement for~$a$ and~$L$ modulo~$S+\dd W_\xx^r$
  \end{algovals}
  \begin{pseudo}
    $s\gets \rho$ \\
    $\eta \gets$ the largest monomial of degree~$s$ \label{line:mainloop} \\
    $B \gets \varnothing$ \\
    $Q \gets \supp([a]_\eta)$ \\
    \kw{while} $Q\setminus B \neq \varnothing$ \\+
      $m\gets $ an element of~$Q\setminus B$ \label{line:pick-m} \\
      \kw{if} $\deg m > s - \rho$ \label{line:if} \\+
      $s\gets s+1$ \\
      \kw{goto} line~\ref{line:mainloop} \\-
      $Q \gets Q \cup \supp([L(m)]_\eta)$ \label{line:augment-Q} \\
      $B\gets B\cup \left\{ m \right\}$ \\-

    \kw{return} $(\eta, B)$ \\
  \end{pseudo}
\end{alg}

\begin{definition}
An \emph{effective confinement}
for~$a \in W_\xx^r$ and a $W_\xx$-linear map~$L\colon W_\xx^r \to W_\xx^r$
is a pair $(\eta,B)$ consisting of a monomial~$\eta$ and of a finite subset~$B\subseteq \monxxr$, and satisfying:
\begin{enumerate}
  \item the support of $[a]_\eta$ is included in~$B$;
  \item the support of $[L(m)]_\eta$ is included in~$B$ for any $m \in B$.
\end{enumerate}
An effective confinement is \emph{free} if the elements of~$B$ are $\bK$-linearly independent modulo~$S+\dd W_\xx^r$.
\end{definition}

\begin{theorem}\label{thm:confinement}
  Algorithm~\ref{algo:confinement} is correct.
  It terminates if $M/\dd M$~is finite-di\-men\-sio\-nal,
  for example if $M$~is holonomic.
  Moreover, if the input parameter~$\rho$ is large enough, then Algorithm~\ref{algo:confinement} outputs a free effective confinement.
\end{theorem}

\begin{proof}
  We first address correctness.
  Consider the sets $B$ and~$Q$ after any iteration of the while loop.
  By construction, we have $B \subseteq Q$,
  $\supp([a]_\eta) \subseteq Q$,
  and $\supp([L(m)]_\eta) \subseteq Q$ for any~$m\in B$.
  If the halting condition $Q\setminus B \neq \varnothing$ of the while loop is reached,
  that is, equivalently, if $Q \subseteq B$ holds at the end of an iteration,
  then we have $B = Q$.
  In conclusion, the returned value~$(\eta,B)$ is an effective confinement.

  As for termination,
  let~$C \subseteq M_\xx^r$ be the set of monomials that are normal forms modulo~$S+\dd W_\xx^r$.
  As a consequence of Definition~\ref{def:nf}, the set of normal forms is the vector space~$\Span_\bK(C)$.
  Moreover, a basis of the quotient space~$M/\dd M$ is formed by the classes modulo $S+\dd W_\xx^r$ of all elements of~$C$,
  so in particular, $C$~is finite by the hypothesis of finite dimension.
  By Corollary~\ref{cor:nf},
  there is therefore some~$\eta_\infty$ such that for any~$\eta \succcurlyeq \eta_\infty$,
  \begin{equation}
    \label{eq:all-confined}
    \supp([a]_\eta) \subseteq C \text{ and } \forall m\in C, \ \supp([L(m)]_\eta) \subseteq C.
  \end{equation}

  Since each iteration of the \emph{while} loop treats a different monomial~$m$,
  and since there are finitely many monomials of degree at most~$s -\rho$,
  the \emph{while} loop terminates.
  It terminates either because~$Q\setminus B = \varnothing$, in which case the algorithm terminates,
  or because~$Q$ contains an element of degree larger that~$s - \rho$, in which case we increase~$s$.
  So, either the algorithm terminates, or $s$~tends to~$\infty$.

  Assume~$s\to\infty$.
  At some point, we will have $s \geq \deg \eta_\infty$,
  so after line~\ref{line:mainloop} is executed,
  we have the inequalities
  $\eta \succcurlyeq \eta_\infty x_1^{s-\deg(\eta_\infty)} \succcurlyeq \eta_\infty$,
  because $\eta$~is the largest monomial of degree~$s$ and by the definition of a monomial order.
  In this circumstances, the set~$Q$ is a subset of~$C$ at every iteration of the main loop, because of~\eqref{eq:all-confined}, and so is~$B$ because of the invariant~$B\subseteq Q$.
  Since~$s \to \infty$, we also reach a point where~$\rho + \deg m \leq s$
  for all~$m\in C$. After this point, $s$~is not increased anymore.
  This contradiction shows that the algorithm terminates.

  If the input~$\rho$ satisfies $\rho \geq \deg \eta_\infty$,
  we have $s \geq \rho \geq \deg \eta_\infty$,
  and by the same reasoning as in the previous paragraph,
  we have again $\eta \succcurlyeq \eta_\infty$.
  This is so during the whole execution of the algorithm.
  Therefore, like in the preceding paragraph,
  we have $B \subseteq Q \subseteq C$ during the execution of the while loop.
  So, the output set~$B$ is a subset of~$C$, which is a free family modulo~$S+\dd W_\xx^r$.
\end{proof}

\begin{example}\label{ex:airyL}
Continuing Examples~\ref{ex:airy-red} and~\ref{ex:airy-irred},
define the $W_{x,y,z}(t)$-linear map
\begin{equation}\label{eq:L-of-example}
L\colon W_{x,y,z}(t) \to W_{x,y,z}(t) \, \text{ by } \, L(m) = \tfrac{1}{2} m(\partial_z -y).
\end{equation}
It will be related to the derivation with respect to~$t$ in Example~\ref{ex:airy-fin}.

We now detail the execution of Algorithm~\ref{algo:confinement} with input
$S = \ann(f)$, the operator $L$ defined in~\eqref{eq:L-of-example}, and parameters $a = 1$ and $\rho = 1$.

The algorithm first sets $s = 1$, $\eta = x$, $B = \varnothing$, and $Q = \{1\}$,
as $a = 1$~is already reduced.
It then picks~$m = 1$, the only possible choice from~$Q \setminus B$,
which does not satisfy the condition on line~\ref{line:if} because~$\deg m = 0 = s - \rho$.
So the algorithm proceeds to line~\ref{line:augment-Q}
and computes the support of $[L(1)]_{x}$, which is~$\{y\}$.
Accordingly, $Q$~is set to~$\{1,y\}$ and $B$~is set to~$\{1\}$.
At the next iteration of the while loop, the algorithm has to pick~$m = y$.
Since $\deg(y) > s - \rho = 0$, the condition on line~\ref{line:if} is satisfied,
thus the algorithm restarts with~$s=2$.
The variable $\eta$ is now set to~$x^2$.
The sets $B$ and~$Q$ are reset to $\varnothing$ and~$\{1\}$, respectively.
Again, the element~$m$ picked at line~\ref{line:pick-m} is~1
and the algorithm proceeds to line~\ref{line:augment-Q}
where it computes $\supp([L(1)]_{x^2})$, which is equal to~$\{y\}$.
The sets $Q$ and~$B$ become $\{1,y\}$ and~$\{1\}$, respectively.
The while loop continues with~$m = y$, the only element in~$Q\setminus B$.
The algorithm therefore computes the support of $[L(y)]_{x^2}$, which is $\{1\}$ by Example~\ref{ex:airy-irred}.
The monomial~1 has already been processed, so no new elements is added to~$Q$,
and $B$~is extended to~$\{1,y\}$.
At this point, $B = Q$ and there are no more monomials to deal with.
So the algorithm terminates and returns $\eta = x^2$ and $B = \{1, y\}$,
which by Theorem~\ref{thm:confinement} is an effective confinement.
\end{example}

\subsection{Comparison with the Griffiths--Dwork reduction}

\label{subsec:GD}

Let~$f \in \bK[\xx]$ be a homogeneous polynomial
and let $M$ be the $W_\xx$-module~$\bK[\xx]  e^{f}$, where $\partial_i$ acts by $\partial_i \cdot e^{f} = \frac{\partial f}{\partial x_i} e^f$.
When~$f$ defines a smooth variety, we can compute in~$M/\dd M$ using the Griffiths--Dwork reduction~\cite{Dwork62,Dwork64,Griffiths-1969-PCR}.
This is usually presented with rational functions in~$\bK[\xx, f^{-1}]$ but the exponential formulation is equivalent (for example, see \parencite{Dimca_1990,Lairez-2016-CPR}).
The module~$M$ admits the presentation
\[ M \simeq \frac{W_\xx}{\sum_i W_\xx (\partial_i - f_i)}, \]
where~$f_i$ denotes the partial derivative $\frac{\partial f}{\partial x_i}$.
(This presentation is what makes the exponential formulation easier in our setting. A presentation of the holonomic $W_\xx$-module $\bK[\xx, f^{-1}]$ is much harder to compute \parencite{OakuTakayama_2001}.)

We briefly present the Griffiths--Dwork reduction and observe that irreducible elements for the Griffiths--Dwork reduction are exactly the irreducible elements for our reduction $\to$.

\paragraph{The Griffiths--Dwork reduction.}

Let $\preccurlyeq_0$ be a monomial order on $\bK[\xx]$,
and, for this monomial ordering, let $G_{0}$ be the minimal Gröbner basis of the polynomial ideal $ I = (f_1,\dots,f_{n})$.
Given a homogeneous polynomial $a \in \bK[\xx]$,
we can compute the remainder~$r$ of the multivariate division of~$a$ by~$G_0$ and the cofactors~$b_1,\dotsc,b_n \in \bK[\xx]$ such that
\begin{equation}
\label{introGD1}
a = r + \sum_{i=1}^{n}{b_if_i}.
\end{equation}
By homogeneity, $\deg b_i = \deg a - \deg f + 1$ (unless~$b_i = 0$).
Then, the rule for the derivative of a product yields
\begin{equation}
\label{introGD2}
ae^{f} = r e^{f} - \underbrace{\sum_{i=1}^{n}{\frac{\partial b_i}{\partial x_i}} }_{\mathclap{\text{of degree $\deg a - \deg f$}}} e^{f} + \underbrace{\sum_{i=1}^{n} \partial_i \cdot \left(b_ie^{f}\right)}_{\in \dd M}.
\end{equation}
The last term is in~$\dd M$, so we ignore it,
and the second term has lower degree than~$a$, so we can apply the same procedure recursively,
which will terminate by induction on the degree.
In the end, we obtain a reduced form $a e^f \equiv r e^f \pmod{\dd M}$ where~$r \in \bK[\xx]$ is irreducible with respect to~$G_0$. These are the irreducible elements for the Griffiths--Dwork reduction.

This reduction is defined for any homogeneous polynomial~$f$, but it enjoys special properties when $f, f_1,\dotsc,f_n$ do not have any non-trivial common zero
in an algebraic closure of~$\bK$. Geometrically, this means that $f$ defines a smooth hypersurface in~$\mathbb{P}^{n-1}(\bK)$.
\textcite{Griffiths-1969-PCR} proved, under this smoothness assumption, that the reduced form of any~$a e^f$ vanishes if and only if~$a e^f \in \dd M$.



\paragraph{Comparison with the reduction $\rightarrow$.}

We consider the reduction rule $\to$ applied to the left ideal~$S$ of~$W_\xx$ generated by~$\partial_i - f_i$, so that~$\bK[\xx]e^f \simeq W_\xx/S$ (Section~\ref{subsec:red}).
Let $\preccurlyeq$ be a monomial order on~$W_\xx$ that eliminates~$\xx$ (see~\eqref{eq:elimination-order})
and agrees with~$\preccurlyeq_0$ on~$\bK[\xx]$.
By following the steps of Buchberger's algorithm, we observe that there is a Gröbner basis~$G$ of~$S$ in which each element is:
\begin{enumerate}
  \item either an element of the form $r - \sum_{i=1}^n b_i \partial_i$, with $r\in G_0$, $b_i \in \bK[\xx]$ of degree $\deg(r) - \deg(f_i)$, and~$r = \sum_{i} b_i f_i$,
  \item or an element of $W_\xx \partial_1 + \dotsb + W_\xx \partial_n$.
\end{enumerate}
We will call such elements respectively of the \emph{first kind} and of the \emph{second kind}.

We now characterize irreducible elements in~$W_\xx$ with respect to~$\to$.
Let~$a \in W_\xx$ be an irreducible element.
Since~$a$ cannot be reduced with $\to_2$, it contains no~$\partial_i$, so it is a polynomial.
Since~$a$ cannot be reduced with $\to_1$, no monomial in~$a$ is divisible by the leading term of an element of~$G$. By considering the elements of the first kind, we see that the monomials of~$a$ are not divisible by the leading monomial of any element of~$G_0$.
So~$a$ is irreducible with respect to~$G_0$.
The converse also holds: if~$a \in \bK[\xx]$ is irreducible with respect to~$G_0$, then~$a$ is irreducible in~$W_\xx$ with respect to~$S$.
In this sense, we can regard Algorithm~\ref{algo:irred-form} as a generalization of the Griffiths--Dwork reduction.

As we observed, this reduction is not enough to compute in~$M/\dd M$, since there may be nonzero irreducible elements in~$S + \dd W_\xx$.
In the case of rational functions, \textcite{Lairez-2016-CPR} gave an algorithm to compute them efficiently.
The algorithm that we have given in Section~\ref{subsec:compirred} behaves differently.
In short, the algorithm in~\parencite{Lairez-2016-CPR} would only consider elements of the second kind with degree~1 in the~$\partial_i$, whereas we consider all elements of the second kind. On the one hand, this seems to give more reduction power, on the other hand the cost of computing them is higher. This indicates room for improvement in future work.

\section{Creative Telescoping by Reduction}
\label{sec:CTbyRed}

In the previous section, we obtained an algorithm for normalizing modulo
derivatives in a holonomic $W_\xx$-module.
In this section, we introduce a parameter~$t$ and differentiation with respect to~$t$.
It would be natural to work with a holonomic~$W_{t,\xx}$-module, but in view of the previous section, we need a finitely presented module
over a Weyl algebra in the derivatives with respect to~$\xx$ only.
This motivates the following context.

We consider the Weyl algebra $W_\xx(t) = \bK(t)\otimes_\bK W_\xx$ (which is just a Weyl algebra over the field~$\bK(t)$),
and a holonomic $W_\xx(t)$-module~$M$ with a compatible derivation~$\partial_t$, that is, a~$\bK$-linear map~$\partial_t\colon M \to M$ such that for any~$a\in W_\xx(t)$ and any~$m\in M$,
\[ \partial_t \cdot a m  = \tfrac{\partial a}{\partial t} m + a \partial_t \cdot m, \]
where $\tfrac{\partial a}{\partial t}$ is the coefficient-wise differentiation in~$W_\xx(t)$.
In other words, $M$ is a~$W_{t,\xx}(t)$-module that is holonomic as a~$W_{\xx}(t)$-module.

It is also convenient to fix a finite presentation~$W_\xx(t)^r/S$ of~$M$ and assume that there is a~$W_\xx(t)$-linear map~$L\colon W_\xx(t)^r \to W_\xx(t)^r$ such that for any~${a\in W_\xx(t)^r}$,
\begin{equation}
  \label{eq:def-prop-L}
  \partial_t \cdot \proj_S(a) = \proj_S \left( \tfrac{\partial a}{\partial t} + L(a) \right),
\end{equation}
where~$\proj_S$ is the canonical map~$W_\xx(t)^r \to M$.
In particular, note that~$S$ is stable under~$\tfrac{\partial}{\partial t} + L$.
From the algorithmic point of view, we represent~$M$ by its finite presentation and the derivation~$\partial_t$ by the $r\times r$ matrix of the endomorphism~$L$.
We explain in Section~\ref{subsec:scalar-extension} how to obtain this setting from a holonomic~$W_{t,\xx}$-module.

Using the algorithm of the previous section,
we aim at describing an algorithm that performs integration with respect to~$x_1,\dotsc,x_n$, in the following sense.
Given~$f\in M$, we want to compute a nonzero operator~$P(t, \partial_t) \in W_t(t)$ such that
\begin{equation}\label{eq:wanted-relation}
  P(t, \partial_t) \cdot f \in \dd M
\end{equation}
with the motivation that
$P(t, \partial_t)$~is then an annihilating operator of the integral of~$f$ with respect to~$x_1,\dotsc,x_n$.
The principle of integration by reduction is described in Section~\ref{sec:integr-algor}
and an algorithm is presented in Section~\ref{subsec:actual-algo}.

\subsection{Integration by reduction}
\label{sec:integr-algor}
We utilize the family of reductions~$[.]_\eta$ defined in Section~\ref{subsec:compirred}.
Let $f$ be an element of $W_\xx(t)^r$, let $\eta$ be some monomial in~$\monxxr$ and
let $(g_i)_{i\geq 0}$ be the sequence in~$W_{\xx}(t)^r$ defined by
\begin{equation}\label{eq:successive-nf}
g_0 = [f]_\eta
\quad\text{and}\quad
g_{i+1} = \tfrac{\partial g_i}{\partial t} +  \left[ L(g_i) \right]_\eta \text{ for all~$i\geq0$.}
\end{equation}

As usual with integration-by-reduction algorithms,
we relate the dependency relations between the reduced forms~$g_i$
to the operators~$P\in W_t(t)$
such that $P \cdot \proj_S(f) \in \dd M$,
which, as is traditional, we call \emph{telescopers} for~$f$.

\begin{lemma}\label{lem:gis}
  For any~$i\geq 0$,
  \begin{equation}\label{eq:g_i-congr}
    \proj_S(g_i) \equiv \partial_t^i\cdot \proj_S(f) \pmod{\dd M}.
  \end{equation}
\end{lemma}

\begin{proof}
  For~$i = 0$, this means that~$[f]_\eta \equiv f \pmod{S + \dd W_\xx(t)^r}$, which holds by construction of~$[.]_\eta$.
  By property of~$[.]_\eta$, again, there is some~$s_i\in S$ and some~$\Delta_i = \sum_j \partial_j a_{i,j} \in \dd W_\xx(t)^r$ such that
  \[ \left[  L(g_i) \right]_\eta =  L(g_i) + s_i + \Delta_i. \]
  Therefore, using the $W_\xx(t)$-linearity of~$\proj_S$ and~$\proj_S(s_i) = 0$, we obtain
  \begin{align*}
    \proj_S(g_{i+1}) &= \proj_S \left( \tfrac{\partial g_i}{\partial t} + L(g_i) \right) + \proj_S(\Delta_i) \\
                                 &= \partial_t \cdot \proj_S(g_i) + \sum_j \partial_j \proj_S(a_{i,j}), && \text{using \eqref{eq:def-prop-L},}\\
                                 &  \equiv \partial_t \cdot \proj_S(g_i) \pmod{\dd M}.
  \end{align*}
  The claim follows by induction on~$i$, using that~$\partial_t \cdot \dd M \subset \dd M$, since~$\partial_t$ commutes with~$\dd$.
\end{proof}

\begin{lemma}\label{lem:telescoper-by-reduction}
  Let $P=\sum_{i=0}^N c_i \partial_t^i \in W_t(t)$.
  \begin{enumerate}
    \item If~$c_0 g_0 + \dotsb + c_N g_N = 0$, then $P \cdot \proj_S(f) \in \dd M$.
    \item\label{it:tel-to-rel} If~$P \cdot \proj_S(f) \in \dd M$, then~$c_0 g_0 + \dotsb + c_N g_N \in S + \dd W_\xx(t)^r$.
    \item If~$\eta$ is large enough and if~$P \cdot \proj_S(f) \in \dd M$, then~$c_0 g_0 + \dotsb + c_N g_N = 0$. (Note that “large enough” depends on~$f$ and~$P$.)
  \end{enumerate}
\end{lemma}

\begin{proof}
  The first assertion follows directly from Lemma~\ref{lem:gis}.
  Conversely, assume that~$P \cdot \proj_S(f) \in \dd M$.
  This implies, again by Lemma~\ref{lem:gis}, that
  \[ \sum_{i=0}^N c_i g_i  \in S + \dd W_\xx(t)^r , \]
  proving the second assertion.
  Now, we observe that, if in addition $\eta$~is large enough, the~$g_i$ are normal forms modulo~$S + \dd W_{\xx}(t)^r$.
  Indeed, by Corollary~\ref{cor:nf}, the~$[L(g_i)]_\eta$ are normal forms;
  and since being a normal form is a condition on the monomial support, it is stable under coefficient-wise differentiation, so the~$\frac{\partial g_i}{\partial t}$ are normal forms, by induction on~$i$.
  So the linear combination~$\sum_i c_i g_i$ is also a normal form modulo $S + \dd W_{\xx}(t)^r$.
  This implies~$\sum_i c_i g_i = 0$.
\end{proof}

\subsection{An algorithm for integrating by reduction}
\label{subsec:actual-algo}

To turn Lemma~\ref{lem:telescoper-by-reduction} into an algorithm to compute a telescoper,
it only remains to find a suitable~$\eta$.
We use the idea of \emph{confinement} (Section~\ref{subsec:guess_eta}).
Using Algorithm~\ref{algo:confinement}, we can compute an effective confinement for $f$ and~$L$.
Recall that this is a monomial~$\eta$ and a finite set~$B$ of monomials in~$M_{\xx,r}$ such that~$\supp([f]_\eta)\subseteq B$
and~$\supp([L(b)]_\eta) \subseteq B$ for all~${b\in B}$.
The following statement explains that reduced forms of successive derivatives with respect to~$t$
therefore lie in the finite-dimensional vector-space~$\Span_{\bK(t)}(B)$.

\begin{lemma}\label{lem:integration-by-reduction-fine}
  Let~$(\eta, B)$ be an effective confinement for~$f$ and~$L$.
  Let $(g_i)_{i\geq 0}$ be the sequence defined by~\eqref{eq:successive-nf}.
  Then, for all~$i \geq 0$, $g_i \in \Span_{\bK(t)}(B)$.
\end{lemma}

\begin{proof}
  By definition of an effective confinement, $[f]_\eta \in \Span_{\bK(t)}(B)$,
  and the space $\Span_{\bK(t)}(B)$ is stable under~$[L(.)]_\eta$.
  Moreover, $\Span_{\bK(t)}(B)$ is stable under~$\frac{\partial}{\partial t}$.
  So the claim follows from the definition of~$g_i$.
\end{proof}

Algorithm~\ref{algo:confinement} and Lemma~\ref{lem:integration-by-reduction-fine} combine into Algorithm~\ref{algo:integration},
whose main properties are provided in the following theorem.

\begin{alg}{Integration using reductions}{integration}
  \textbf{Input:}
  \begin{algovals}
    \item a holonomic module~$W_\xx(t)^r/S$
    \item a derivation map~$\partial_t\colon W_\xx(t)^r/ S \rightarrow W_\xx(t)^r/S$
      given by the matrix of an endomorphism $L$, as in \eqref{eq:def-prop-L}.
    \item an element~$f \in W_\xx(t)^r$
    \item an integer~$\rho \geq 0$
    \end{algovals}
  \textbf{Output:}
  \begin{algovals}
    \item $P = c_0 + \dots + c_N\partial_t^N$ such that $c_i\in\bK(t)$, $c_N \neq 0$ and $P\cdot \proj_S(f) \in \dd M$
  \end{algovals}
  \begin{pseudo}
    $(\eta, B) \gets$ an effective confinement obtained from~$(S, f, L, \rho)$ by Algorithm~\ref{algo:confinement} \\
    $g_0 \gets [f]_\eta$ \\
    $N \gets 0$ \\
    \kw{while} $g_0, \dotsc, g_N$ are linearly independent over~$\bK(t)$ \\+
      $g_{N+1} \gets \frac{\partial g}{\partial t} + \left[ L(g) \right]_\eta$ \\
      $N \gets N+1$ \\-
    \kw{return} $c_0 + \dots + c_N\partial_t^N$ s.t.~$c_0 g_0 + \dotsb +c_Ng_N = 0$, $c_i\in\bK(t)$ and $c_N\neq 0$. \\
  \end{pseudo}
\end{alg}

\begin{theorem}\label{thm:algo-integration-correct}
  Algorithm~\ref{algo:integration} is correct and terminates.
  Moreover, if $\rho$~is large enough, then it outputs a minimal telescoper for the input.
\end{theorem}

\begin{proof}
  Correctness follows from Lemma~\ref{lem:telescoper-by-reduction}.
  As to termination, it follows from Lemma~\ref{lem:integration-by-reduction-fine}:
  because the set~$B$ is finite, the infinite family of elements $g_0,g_1,\dots$ is linearly dependent,
  so the main loop terminates for some~$N$ less than or equal to the cardinality of~$B$.

  As for the minimality,
  it is clear that the algorithm outputs a non-trivial relation~$c_0 g_0 + \dotsb + c_Ng_N = 0$ with minimal possible~$N$ among those available for the sequence~$(g_i)_{i\geq0}$.
  Besides,
  consider any telescoper $P = c_0+\dotsc+c_\Omega\partial_t^\Omega$.
  By point~\ref{it:tel-to-rel} of Lemma~\ref{lem:telescoper-by-reduction}, we have
  \begin{equation}
    \label{eq:min-tel-vanishes}
    c_0 g_0 + \dotsb + c_\Omega g_\Omega \in S + \dd W_\xx(t)^r.
  \end{equation}
  Assume that $\rho$~is large enough, in the sense of Theorem~\ref{thm:confinement}, so that the confinement is free,
  meaning that the elements of~$B$ are independent modulo~$S+\dd W_\xx^r$.
  The linear combination in~\eqref{eq:min-tel-vanishes}
  is a linear combination of elements of~$B$, so it must be zero: $c_0 g_0 + \dotsb + c_\Omega g_\Omega = 0$.
  So Algorithm~\ref{algo:integration} will output a relation for an~$N$ that is at most the minimal order of telescopers.
\end{proof}

\begin{remark}
  Algorithm~\ref{algo:integration} can be modified to compute a system of linear differential equations
  satisfied by an integral depending on multiple parameters $t_1,\dots,t_p$. These parameters can also be associated
  to other Ore operators~\cite{CHYZAK1998187} than the differentiation
  provided they define a map on $W_\xx^r(t_1,\dots,t_p)$.
\end{remark}

\subsection{Scalar extension}\label{subsec:scalar-extension}

The input data structure of the integration algorithm (Algorithm~\ref{algo:integration})
is unusual, with the derivation with respect to the parameter which is not part of the holonomic context.
This section explains how to obtain such a representation from a holonomic~$W_{t,\xx}$-module.
The formulation in full generality is a bit technical.
In practice, this amounts to a single Gröbner basis computation over~$W_{t,\xx}(t)$, and, in some examples, this can be very simple, see Example~\ref{ex:airy-fin} below.

Let~$H$ be a holonomic $W_{t,\xx}$-module
and let~$M = \bK(t) \otimes_{\bK[t]} H$.
The space~$M$ is a $W_\xx(t)$-module in a natural way.
Moreover, we can define a derivation~$\partial_t$ by
\[ \partial_t \cdot (a\otimes m) = \tfrac{\partial a}{\partial t} \otimes m + a \otimes (\partial_t \cdot m). \]
This derivation commutes with the action of the~$\partial_i$.

In this section, we aim to compute~$M$ from~$H$, so that we can apply the integration algorithm described in Section~\ref{subsec:actual-algo}.
Let us first make explicit what we mean.
We assume that $H$~is given by a finite presentation,
that is, $H = W_{t,\xx}^s / J$ for some~$s \geq 0$ and some submodule~$J\subseteq W_{t,\xx}^s$
given by a finite set of generators.
Computing~$M$ means computing a finite presentation~$M \simeq W_\xx(t)^r / S$,
and a $W_\xx(t)$-linear map~$L\colon W_\xx(t)^r \to W_\xx(t)^r$ such that~\eqref{eq:def-prop-L} holds.
It is not obvious that such a finite presentation exists because~$M$ does not have any obvious finite set of generators.
However, this existence is implied by the holonomicity of~$M$.
Here, we give a proof based on restriction of D-modules.

\begin{lemma}\label{lem:gen-holonomic}
  If $H$~is a $W_{t,\xx}$-holonomic module, then $M = \bK(t) \otimes_{\bK[t]} H$ is a $W_\xx(t)$-holonomic module.
\end{lemma}

The statement is similar in nature to the well-known statement that “holonomic implies D-finite”.

\begin{proof}
  Let~$\xi$ be a new variable.
  Consider the field~$\bL = \bK(\xi)$.
  Introduce the ${(1+n)}$th Weyl algebra with coefficients in~$\bL$, which we denote~$W_{t,\xx}(\xi)$.
  Consider as well the~$W_{t,\xx}(\xi)$-module~$H' = \bL \otimes_\bK H$.
  This scalar extension of the base field preserves holonomicity.
  So $H'$~is holonomic.
  Consider now the embedding map~$F\colon \bL^{n} \to \bL^{1+n}$ defined by
  \[ (x_1,\dotsc,x_n) \mapsto (\xi, x_1, \dotsc, x_n). \]
  The inverse image by~$F$ is a functor, denoted~$F^*$, that maps~$W_{t,\xx}(\xi)$-modules to~$W_{\xx}(\xi)$-modules
  \cite[Chapter 14]{Coutinho_1995}.
  For the particular choice of~$F$ above, the inverse image by~$F$ is also known as the restriction along the hyperplane~$\{t=\xi\}$.
  We do not need the details of this construction, we just need to know that
  the inverse image of~$H'$ by~$F$, denoted~$F^* H'$, is:
  \begin{itemize}
  \item a $W_\xx(\xi)$-module, which is by \cite[construction of \S14.1 and~\S14.2]{Coutinho_1995},
  \item holonomic as a $W_\xx(\xi)$-module, which is by \cite[Theorem~18.1.4]{Coutinho_1995},
  \item isomorphic to $H' / (t-\xi) H'$, which is obtained as a suitable variant of \cite[\S15.1]{Coutinho_1995},
    by making~$Y=t$ in that reference before specializing at~$\xi$ instead of~$0$.
  \end{itemize}
  In particular, we have
  \[F^* H'  \simeq H'/(t-\xi)H' \simeq \bL[t]/(t-\xi) \otimes_{\bL[t]} H'. \]
  Next, we check that
  \begin{align*}
    H' &= \bL \otimes_\bK H, && \text{by definition}\\
        &\simeq \bL \otimes_\bK (\bK[t] \otimes_{\bK[t]} H), &&\text{because $H$ is a $\bK[t]$-module}\\
        &\simeq \bL[t] \otimes_{\bK[t]} H, &&\text{by associativity of $\otimes$},
  \end{align*}
  and therefore, using associativity of $\otimes$ again,
  \[ F^* H'  \simeq \bL[t]/(t-\xi) \otimes_{\bK[t]} H. \]
  Finally, we observe the isomorphism $\bK(t) \simeq \bL[t]/(t - \xi)$ as $\bK[t]$-algebras under the map~$f(t) \mapsto f(\xi)$, so
  we obtain $F^* H' \simeq \bK(t) \otimes_{\bK[t]} H = M$, by definition of~$M$. Since~$F^*H'$ is holonomic, this gives the claim.
\end{proof}

We now describe an algorithm for computing~$M$.
Recall that $H = W_{t,\xx}^s/J$, so $M \simeq W_{t,\xx}(t)^s/ J(t)$ where~$J(t)$ is the submodule $\bK(t)\otimes_{\bK[t]} J$ of~$W_{t,\xx}(t)^s$ generated by~$J$.
We can compute normal forms in $M$ using Gröbner bases in~$W_{t,\xx}(t)^s$ after fixing a monomial order on all monomials
$\xx^\alpha \dd^\beta \partial_t^k e_i$ \parencite[e.g.,][]{CHYZAK1998187}.
We choose a monomial order that eliminates~$\partial_t$, that is, any monomial order such that
\[ k < k' \Rightarrow \xx^\aalpha \dd_\xx^\bbeta \partial_t^k e_i \prec \xx^{\aalpha'} \dd_\xx^{\bbeta'} \partial_t^{k'} e_{i'}. \]
Let~$G$ denote a Gröbner basis of~$J(t)$ for such an elimination order.

As a~$W_\xx(t)$-module, $W_{t,\xx}(t)^s$ is generated
by the set
\begin{equation*}
\left\{ \partial_t^i e_j \st i \geq 0, \ 1\leq j \leq s \right\}.
\end{equation*}
So $M$ is generated by the image of this set.
This is an infinite family, but, since $M$~is $W_\xx(t)$-holonomic,
$M$ is actually a Noetherian $W_\xx(t)$-module, finitely generated in particular.
To describe~$M$, we need to find an explicit finite generating set
and the module of relations between the generators.

For $a \in W_{t,\xx}(t)^s$, let~$\ind(a)$ denote the degree of~$a$ with respect to~$\partial_t$,
which we will call the index of~$a$.
In other words, this is the smallest integer~$k \geq 0$ such that
$a$ is in the sub-$W_\xx(t)$-module generated by
\[ B_k = \left\{ \partial_t^i e_j \st 0 \leq i \leq k, \ 1\leq j \leq s \right\}. \]
Moreover, for $a \in W_{t,\xx}(t)^s$, let~$\ind_{J(t)}(a)$ denote
\begin{equation}
  \label{eq:ind-mod}
  \ind_{J(t)}(a) = \min \left\{ \ind(b) \st b \in W_{t,\xx}(t)^s \text{ and } a \equiv b \pmod{J(t)} \right\}.
\end{equation}
Given~$a$, we can compute~$\ind_{J(t)}(a)$ using the Gröbner basis~$G$:
\[ \ind_{J(t)}(a) = \ind \left( \LRem(a, G) \right). \]
Indeed:
we have $a \equiv \LRem(b, G)$ if~$a \equiv b \pmod{J(t)}$ and
the elimination property shows~$\ind \left( \LRem(b, G) \right) \leq \ind(b)$,
so that $\ind(b)$~can be replaced with $\ind(\LRem(b, G))$ in~\eqref{eq:ind-mod};
then the Gröbner basis property shows $\LRem(a, G) = \LRem(b, G)$ if~$a \equiv b \pmod{J(t)}$.

\begin{lemma}\label{lem:generating-set-extension}
  There is~$\ell \geq 0$ such that $\ind_{J(t)}(\partial_t^{\ell+1} e_i) \leq \ell$ for any~$1 \leq i \leq s$.
  Moreover, for any such~$\ell$:
  \begin{enumerate}
  \item $M$~is generated as a $W_\xx(t)$-module by the image in it of~$B_\ell$,
  \item $\ind_{J(t)}(a) \leq \ell$ for any~$a\in W_{t,\xx}(t)^s$.
  \end{enumerate}
\end{lemma}

\begin{proof}
Since $M$~is Noetherian,
the increasing sequence of the $W_\xx(t)$-modules~$S_k$ generated by the images of the~$B_k$ in~$M$ is stationary:
there exists~$\ell\geq0$ for which the $W_\xx(t)$-module~$S_\ell$ contains all the~$S_k$ for~$k\geq0$,
and is therefore equal to~$M$.
For such an integer~$\ell$, any~$k > \ell$, and any~$j$,
the image of~$\partial_t^k e_j$ in~$M$ is in~$S_k$, therefore in~$S_\ell=M$.
Consequently, there exist coefficients~$c_{h,i} \in W_\xx(t)$ satisfying
\[ \partial_t^k e_j \equiv \sum_{h \leq \ell, \ i} c_{h,i} \partial_t^h e_i \pmod{J(t)} . \]
By the definition~\eqref{eq:ind-mod}, $\ind_{J(t)}(\partial_t^k e_j)$~is less than or equal to the index of the right-hand side,
which by construction is less than or equal to~$\ell$.
We obtain that $\ell$~is a uniform bound on all~$\ind_{J(t)}(\partial_t^k e_j)$.
This proves in particular the first part of the statement, on the existence of~$\ell$.
For the second part, we fix such an~$\ell$.
We have already proved~$M = S_\ell$, which is the first itemized statement.
We have also already proved, for any~$k > \ell$ and any~$j$,
the existence of some~$r_{k,j}$ of index at most~$\ell$ such that
$\partial_t^k e_j \equiv r_{k,j} \pmod{J(t)}$.
This also holds by the definition of the index for~$k \leq \ell$.
Now, any~$a \in W_{t,\xx}^s$ writes in the form~$\sum_{k,j} c_{k,j} \partial_t^k e_j$ for coefficients~$c_{k,j} \in W_\xx(t)$.
Taking linear combinations of congruences modulo the $W_\xx(t)$-module~$J(t)$,
we obtain $a \equiv \sum_{k,j} c_{k,j} r_{k,j} \pmod{J(t)}$,
and as a consequence,
\[ \ind_{J(t)}(a) \leq \ind\biggl(\sum_{k,j} c_{k,j} r_{k,j}\biggr) \leq \ell , \]
where the first inequality is by~\eqref{eq:ind-mod} and the second by the definition of the index as a degree.
We have proved the second itemized statement.
\end{proof}

An algorithm for computing the smallest~$\ell$ as in the statement above follows directly from~\eqref{eq:ind-mod},
simply by testing increasing values of~$\ell$.

Now that we have a finite generating set for~$M$, it remains to characterize the relations between the generators.
To this end,
for the rest of the section we fix~$\ell$ as provided by Lemma~\ref{lem:generating-set-extension}
and we let $J_\ell$ be the~sub-$W_{\xx}(t)$-module of~$W_{t,\xx}(t)^s$ generated by
\[ \left\{ \partial_t^k g \st g \in G \text{ and } k + \ind(g) \leq \ell \right\}. \]
It is, by construction, a submodule of~$ W_{\xx}(t) B_\ell$.

\begin{lemma}\label{lem:P-to-M-iso}
  The inclusion~$W_\xx(t) B_\ell\to W_{t,\xx}(t)^s$ induces an isomorphism
  \[ M \simeq \frac{ W_{\xx}(t) B_\ell }{ J_\ell }, \]
  with inverse induced
  by the map~$W_{t,\xx}(t)^s \to W_\xx(t) B_\ell$ given by $a \mapsto \LRem(a, G)$.
\end{lemma}

\begin{proof}
  First, $J_\ell \subseteq J(t)$, so  the inclusion~$W_\xx(t) B_\ell\to W_{t,\xx}(t)^s$
  induces a morphism of~$W_\xx(t)$-modules
  \[ \phi\colon W_\xx(t) B_\ell / J_\ell \to W_{t,\xx}(t)^s / J(t). \]
  Next, the $\bK(t)$-linear map $a \in W_{t,\xx}(t)^s  \mapsto \LRem(a, G)$ has values in  $W_\xx(t) B_\ell$,
  because $\ind( \LRem(a, G) ) = \ind_{J(t)}(a) \leq \ell$, by Lemma~\ref{lem:generating-set-extension}.
  This map vanishes on~$J(t)$, because~$G$ is a Gröbner basis of~$J(t)$, so it induces a $\bK(t)$-linear map
  \[ \psi\colon W_{t,\xx}(t)^s / J(t) \to W_\xx(t) B_\ell / J_\ell. \]
  The maps $\phi\circ \psi$ and~$\psi \circ \phi$ are both induced by~$a\mapsto \LRem(a, G)$.
  The first is the identity on~$W_{t,\xx}(t)^s / J(t)$
  because for all~$a \in W_{t,\xx}^s$, $a \equiv \LRem(a, G) \pmod{J(t)}$ as a property of the Gröbner basis~$G$.
  The second is the identity on~$W_\xx(t) B_\ell / J_\ell$ because of the elimination property:
  the computation of~$\LRem(a, G)$
  only involves multiples of~$G$ of index at most~$\ind(a)$, which are all in~$J_\ell$,
  so that for all~$a$ of index at most~$\ell$, $a \equiv \LRem(a, G) \pmod{J_\ell}$.
  This shows that $\phi$~is an isomorphism.
\end{proof}

At this point we are able to define the wanted dimension~$r$ and module~$S$.
In view of Lemma~\ref{lem:P-to-M-iso}, set $r$ to~$(\ell+1) s$,
so as to have a trivial isomorphism~$W_\xx(t)B_\ell \simeq W_\xx(t)^r$.
Call~$S$ the image of the submodule~$J_\ell$ of~$W_\xx(t)B_\ell$ under this isomorphism,
so that, summarizing,
\begin{equation*}
\frac{W_{t,\xx}(t)^s}{J(t)} \simeq M \simeq \frac{ W_{\xx}(t) B_\ell }{ J_\ell } \simeq \frac{W_\xx(t)^r}{S} .
\end{equation*}

It remains to describe an endomorphism~$L$ of~$W_\xx(t)^r$ such that
\[ \partial_t \cdot \proj_S(a) = \proj_S \left( \tfrac{\partial a}{\partial t} + L(a) \right), \]
for any~$a\in W_\xx(t)^r$.
Introduce the canonical maps $\proj_{J(t)}$ and~$\proj_S$ to the relevant quotients.
Recall that~$\partial_t$ is defined for any $h$ in $W_{t,\xx}(t)^s / J(t)$ by left-multiplication by~$\partial_t$:
\[ \partial_t \cdot \proj_{J(t)}(h) = \proj_{J(t)}(\partial_t h). \]
Therefore, the isomorphism of Lemma~\ref{lem:P-to-M-iso} transfers~$\partial_t$ on $W_{\xx}(t) B_\ell / J_\ell$ as
\begin{equation}\label{eq:dt-pr-Jl}
\partial_t \cdot \proj_{J_\ell}(a) = \proj_{J_\ell} \left( \LRem( \partial_t a, G) \right).
\end{equation}
Lastly, take $a \in W_\xx(t) B_\ell$ and write it $a = \sum_{i=1}^n a_i e_i$.
The Leibniz rule in~$W_{t,\xx}^s$ gives $\partial_t a = \frac{\partial a}{\partial t} + \sum_{i=1}^n a_i \partial_t e_i$.
By linearity of~$\LRem$ and
since $W_\xx(t) B_\ell$~is stable under the coefficient-wise differentiation~$\frac{\partial}{\partial t}$, we obtain
\begin{multline*}
\LRem(\partial_t a, G) = \LRem(\tfrac{\partial a}{\partial t}, G) + \LRem\biggl(\sum_{i=1}^n a_i \partial_t e_i, G\biggr) \\
= \tfrac{\partial a}{\partial t} + \LRem\biggl(\sum_{i=1}^n a_i \partial_t e_i, G\biggr) + h
\end{multline*}
for some~$h \in J_\ell$, then,
upon applying~$\proj_{J_\ell}$ and combining with~\eqref{eq:dt-pr-Jl},
\begin{equation*}
\partial_t \cdot \proj_{J_\ell}(a) =  \proj_{J_\ell} \left( \tfrac{\partial a}{\partial t} + \LRem\biggl(\sum_{i=1}^n a_i \partial_t e_i, G\biggr) \right) .
\end{equation*}
So, the endomorphism~$L$ we want is obtained by transferring
the endomorphism
\[ \sum_{i=1}^n a_i e_i \mapsto \LRem\biggl( \sum_{i=1}^n a_i \partial_t e_i, G \biggr) \]
of~$W_\xx(t)B_\ell$ to an endomorphism of~$W_\xx(t)^r$.

\begin{example}\label{ex:airy-fin}
  With the goal to continue Examples \ref{ex:airy-red}, \ref{ex:airy-irred}, and~\ref{ex:airyL},
  we consider with more generality a function~$f = \exp(q(x,y,z,t))$ where~$q$ is some polynomial.
  Consider the~$W_{t,x,y,z}$-module $H = W_{t,x,y,z}/J$ where $J$ is generated by
  \begin{equation}
    \label{eq:airygb}
    \partial_x - \tfrac{\partial q}{\partial x},
    \partial_y - \tfrac{\partial q}{\partial y},
    \partial_z - \tfrac{\partial q}{\partial z},
    \partial_t - \tfrac{\partial q}{\partial t},
  \end{equation}
  so that~$H \simeq \mathbb{Q}[t,x,y,z] f$.
  The module~$M = \bK(t) \otimes_{\bK[t]} H$,
  viewed as a $W_{x,y,z}(t)$-module after forgetting the differential structure with respect to~$t$,
  is the same as the module~$M$ defined in Example~\ref{ex:airy-red}.

  This case is much easier than the general case addressed above. First, $s=1$, which simplifies the notation; second,  the $\ell$ in Lemma~\ref{lem:generating-set-extension} is~0; therefore~$r= (\ell+1)s = 1$ too.
  Since~$B_0 = \left\{ 1 \right\}$, this means that~1 generates~$M$ as a~$W_{x,y,z}(t)$-module, there is no need to consider derivatives~$\partial_t^k$.
  The generators~\eqref{eq:airygb} form a Gröbner basis~$G$ of~$J(t)$ in~$W_{t,x,y,z}(t)$.
  So~$J_0$ is generated in~$W_{t,x,y,z}(t)$ by
  \[ \left\{
      \partial_x - \tfrac{\partial q}{\partial x},
      \partial_y - \tfrac{\partial q}{\partial y},
      \partial_z - \tfrac{\partial q}{\partial z}, \right\}.
  \]

  In this case, Lemma~\ref{lem:P-to-M-iso} states that the inclusion~$W_{x,y,z}(t) \to W_{t,x,y,z}(t)$ induces an isomorphism
  \[ M \simeq W_{x,y,z}(t) / J_0, \]
  which coincides with the description in Example~\ref{ex:airy-red}.

  It remains to describe the derivation~$\partial_t$ on~$M$.
  In this case, this is easy: $M \simeq \mathbb{Q}(t)[x,y,z] f$ as a~$W_{x,y,z}(t)$-module and the derivation with respect to~$t$ is given by
  \[ \partial_t \cdot a f = \tfrac{\partial a}{\partial t} f + a \tfrac{\partial f}{\partial t} = \bigl(\tfrac{\partial a}{\partial t} + a \tfrac{\partial q}{\partial t}\bigr) f, \]
  for any~$a\in \mathbb{Q}(t)[x,y,z]$, and this generalizes to~$a \in W_{x,y,z}(t)$, since~$\partial_t$ commutes with~$\partial_x$, $\partial_y$ and~$\partial_z$.
  Putting apart the derivation~$\frac{\partial a}{\partial t}$ of the coefficient~$a$, we obtain the linear part of the derivation, $L(a) = a \tfrac{\partial q}{\partial t}$.
  In the specific case of Example~\ref{ex:airy-red}, we have
  $q = \tfrac13 (x^3+y^3) - x(t+2z) - y(t+z)$,
  and $L(a)$~may be simplified to the expression~\eqref{eq:L-of-example} given in Example~\ref{ex:airyL} after reduction of the monomials~$xz$ and~$yz$ in~$\tfrac{\partial q}{\partial t}$.
  As to the general case, we can understand the map~$L$ as
  \[ L(a) = \operatorname{LRem}(a \partial_t, G), \]
  because~$G$ contains the operator~$\partial_t - \tfrac{\partial q}{\partial t}$.
\end{example}

\section{Implementation}
\label{sec:implementation}

This section outlines specific algorithmic and implementation choices  made in our Julia implementation of the algorithms presented in this paper.
We begin in Section~\ref{subsec:general-comment-impl} with a general presentation of our package.
In Section~\ref{sec:efficiency}  we show how memoization can be used
to efficiently compute the image of the reduction map~$[.]_\eta$ on a finite-dimensional space.
Lastly, we present in Section~\ref{sec:crt-e/i} a modified version of Algorithm~\ref{algo:integration}
that utilizes the evaluation/interpolation paradigm
to avoid the growth of intermediate coefficients.
This is of particular importance when reducing operators with coefficients in~$\bK(t)$ using the reduction~$[.]_\eta$.
Additionally, we present in the same section the use of a tracer for computing a basis of $\Irr_{\preccurlyeq\eta}$.

\subsection{General presentation of the package}
\label{subsec:general-comment-impl}

We have implemented our algorithms in Julia.
This is available as the package \verb+MultivariateCreativeTelescoping.jl+\footnote{See \url{https://hbrochet.github.io/MultivariateCreativeTelescoping.jl/}.}.
The package includes an implementation of Weyl algebras
in which operators have a sparse representation by a pair of vectors,
one for exponents and one for the corresponding coefficients.
The currently supported coefficient fields are the field~$\bQ$ of rational numbers, the finite fields~$\bF_p$ with $p\leq 2^{31}$,
and extensions of those fields with symbolic parameters.
When such parameters are present, the implementation interfaces with FLINT
for defining and manipulating commutative polynomials,
by means of the Julia packages \verb+AbstractAlgebra.jl+ and~\verb+Nemo.jl+.
Our package also provides an implementation of non-commutative generalizations of algorithms
for computing Gröbner bases: Buchberger's algorithm (see e.g.~\cite{Bahloul-2020-GBD}), the F4 algorithm~\cite{FAUGERE199961},
and the F5 algorithm~\cite{10.1145/2442829.2442879, LAIREZ2024102275}.
The F4 implementation seems to be the most efficient one for grevlex orders
while the F5 implementation seems to be the most efficient one for block and lexicographical orders.


\subsection{Efficient computation of the sequence $(g_i)_i$ using memoization}
\label{sec:efficiency}

The confinement~$(\eta,B)$ required by Algorithm~\ref{algo:integration} is constructed so that the sequence $(g_i)_i$
defined in~\eqref{eq:successive-nf} is contained in $\Span_{\bK(t)}(B)$.
Properties of the reduction~$[.]_\eta$ imply a refined formula:
after decomposing~$g_i$ as a sum $\sum_{m\in B} a_m m$ with $a_m\in \bK(t)$, $g_{i+1}$~can be obtained by
\begin{equation}\label{eq:gi-efficient}
  g_{i+1} = \sum_{m\in B}\left(\frac{\partial a_i}{\partial t}m + a_m[L(m)]_\eta\right).
\end{equation}
Early in the execution of Algorithm~\ref{algo:integration},
when it computes the confinement $(\eta,B)$ by Algorithm~\ref{algo:confinement},
the image $[L(m)]_\eta$ of every monomial $m\in B$ has to be computed.
We therefore choose to store these images in memory
to allow for a more efficient computation of the sequence~$(g_i)_i$
by using~\eqref{eq:gi-efficient} at a later stage.




\subsection{Modular methods and evaluation/interpolation}\label{sec:crt-e/i}

When, as in this subsection, the base field~$\bK$ is fixed to be~$\bQ$,
evaluation/in\-ter\-po\-la\-tion and modular techniques have long been used in computer algebra
to speed up calculations
\cite{Strassen-1973-VD,Zippel-1979-PAS,Kaltofen-1988-GCD,KaltofenTrager-1990-CPG}.
By making random choices of primes for modular reductions and of integers for variable evaluations,
the resulting algorithms are invariably probabilistic.
When no reasonable a~priori bound on the output size is known,
which is our situation,
the correctness of algorithms is not proven,
and outputs are only extremely plausible.
In Section~\ref{sec:global-mod-eval-interp},
we describe how we reconcile this commutative algebra approach with our differential setting.
More specific ideas have also already been developed in the context of Gröbner bases calculations,
with the additional notion of tracer~\cite{Traverso89}:
in classical contexts, the probability of success of certain parts of the calculation
can be made as close to one as wished,
but the final reconstruction steps make algorithms retain some incompleteness of correctness.
We make use of the same idea in our Algorithm~\ref{algo:basis-E-lambda},
and we comment on this in Section~\ref{sec:using-a-tracer}.

\subsubsection{Global modular evaluation-interpolation scheme}
\label{sec:global-mod-eval-interp}

We first recall the principle of the evaluation/interpolation
paradigm and then we present a modified version of Algorithm~\ref{algo:integration} that incorporates this paradigm.

\paragraph{The principle.}
Let $r_1,\dots,r_\ell \in \bQ(t)$ be rational functions and let $F\colon \bQ(t)^\ell \rightarrow \bQ(t)$ be
a function computable using only additions, multiplications, and inversions.
The rational function $F(r_1,\dots,r_\ell)$ can be reconstructed as an element of~$\bQ(t)$ from its evaluations
$F(r_1(a_i),\dots,r_\ell(a_i)) \bmod p_j$ at several points~$a_i\in\bF_{p_j}$ and for several prime integers~$p_j$
in three steps:
\begin{enumerate}
  \item for each $j$, reconstruct $F(r_1,\dots,r_\ell) \bmod p_j$ in $\bF_{p_j}(t)$
   by Cauchy interpolation~\cite[Chapter 5.8]{Gathen1999}, \\
  \item reconstruct $F(r_1,\dots,r_\ell) \bmod N$ in $\bF_N(t)$ where $N = \prod_j p_j$
  by the Chinese remainder theorem~\cite[Chapter 5.4]{Gathen1999}, \\
  \item lift the integer coefficients of $F(r_1,\dots,r_\ell) \bmod N$ in~$\bQ$
    by rational reconstruction~\cite[Chapter 5.10]{Gathen1999}.
\end{enumerate}
The first (resp. third) step requires a bound on the degree in $t$ of the result (resp. on the size of its coefficients)
to determine the number of evaluations needed to ensure correctness.
Since we do not have access to such bounds, we rely instead on a probabilistic approach:
after successfully obtaining a result with a certain number of  evaluations in step 1 (resp. 3),
we compute one additional evaluation and check consistency with the previously obtained result.

Finally, this method can fail if some bad evaluation points or some bad primes are chosen.
However, these situations are very rare, provided the prime numbers are sufficiently large
and the evaluation points are selected uniformly at random over~$\bF_p$.

\newgeometry{left=3cm,right=3cm}
\begin{alg}{Integration with evaluation/interpolation}{integration-crt-e-i}
  \textbf{Input:}
  \begin{algovals}
    \item a holonomic module~$W_\xx(t)^r/S$
    \item a derivation map~$\partial_t\colon W_\xx(t)^r/ S \rightarrow W_\xx(t)^r/S$
      given by the matrix of an endomorphism $L$, as in~\eqref{eq:def-prop-L}
    \item an element~$f \in W_\xx(t)^r$
    \item an integer~$\rho \geq 0$
    \end{algovals}
  \textbf{Output:}
  \begin{algovals}
    \item $P = c_0 + \dots + c_N\partial_t^N$ such that $c_i\in\bQ(t)$, $c_N \neq 0$ and $P\cdot \proj_S(f) \in \dd M$
  \end{algovals}
  \begin{pseudo}
    \kw{for} large random prime numbers $p = p_1,p_2,\dots$ \\+
    $ \bar{L} \gets $ the endomorphism of $\bF_p\otimes_\bZ W_\xx^r(t)$ induced by $L$ via the reduction modulo $p$ \\
      \kw{for} random numbers  $a = a_1,a_2,\dots$ in $\bF_p$\\+
        $ \tilde{S} \gets $ image of $S$ under evaluation at $t=a$ and reduction modulo $p$ \\
        $ \tilde{L} \gets $ the endomorphism of $\bF_p\otimes_\bZ W_\xx^r$ induced by $\bar{L}$ via the same evaluation \\
        $(\eta, B) \gets$ an effective confinement obtained from~$(\tilde{S}, \tilde{f}, \tilde{L}, \rho)$ by Algorithm~\ref{algo:confinement} over $\bK = \bF_p$\\
        $\tilde{g}_0 \gets [\tilde{f}]_\eta$ where the reduction is over $\bK = \bF_p$\\
        store $\tilde{g}_0$ as well as the images $[\tilde{L}(m)]_\eta$ that have been computed at line $6$ for all $m\in B$  \\-
        interpolate the coefficients of $\bar{g}_0$ and of each $[\bar{L}(m)]_\eta$ by elements of $\bF_p(t)$ \\
      $N \gets 0$ \\
      \kw{while} $\bar{g}_0, \dotsc, \bar{g}_N$ are linearly independent over~$\bF_p(t)$ \\+
      $\bar{g}_{N+1} \gets \frac{\partial \bar{g}_N}{\partial t} + \left[ \bar{L}(\bar{g}_N) \right]_\eta$  where the reduction is over $\bK = \bF_p(t)$\\
      $N \gets N+1$ \\-
      store coefficients of the minimal non-trivial relation $\bar{c}_0 \bar{g}_0 + \dotsb + \bar{c}_{N}\bar{g}_N = 0$  for $\bar{c}_i\in\bF_p(t)$\\-
    reconstruct the coefficients~$c_0,\dots,c_N$ in $\bQ(t)$ from their values in the $\bF_{p_j}(t)$ \\
    \kw{return} $c_0 + \dots + c_N\partial_t^N$ \\
  \end{pseudo}
\end{alg}
\restoregeometry

\paragraph{Implementation.}
The evaluation/interpolation scheme described above cannot be directly applied to Algorithm~\ref{algo:integration}
as it involves not only additions, multiplications and inversions, but also differentiations.
Indeed, recall that the sequence~$(g_i)_i$ is defined for $\eta\in\monxxr$ by $g_0=[f]_\eta$ and
\[
g_{i+1} = \frac{\partial g_i}{\partial t} + [L(g_i)]_\eta
\]
where $\partial g/\partial t$ denotes the coefficient-wise differentiation of $g$ in the basis $\monxxr$.
This differentiation does not commute with the evaluation of $t$,
which prevents the sequence $(g_i)_i$ from being computed by evaluation/interpolation.

Algorithm~\ref{algo:integration-crt-e-i} computes the matrix of the
$\bK(t)$-linear map $[L(.)]_\eta$ using eval\-u\-a\-tion/interpolation and uses it to
find a linear relation among the elements of the sequence $(g_i)_i$.
 The reconstruction of the coefficients from $\bF_p(t)$ to $\bQ(t)$ is delayed to the end of the algorithm.
We adopt the following conventions: the projection of an element $g\in W_\xx(t)^r$
 in $\bF_p\otimes_\bZ W_\xx^r$ by evaluation at a point $a$ and reduction modulo $p$ is stored in a variable $\tilde{g}$
 and its projection in $\bF_p\otimes_\bZ W_\xx^r(t)$ by reduction modulo $p$ is stored in a variable $\bar{g}$.

\subsubsection{Computation of the vector space $\Irr_{\preccurlyeq \eta}$ using a tracer}
\label{sec:using-a-tracer}

In Algorithm~\ref{algo:basis-E-lambda}, a generating set
of the vector space~$\Irr_{\preccurlyeq\eta}$ is computed by reducing for each $\eta'\preccurlyeq\eta$
a term of the form $mg - \partial_i w$ satisfying $\eta' = \lm(mg)=\lm(\partial_i w)$.
However, not all such terms contribute to an increase in the dimension of the space, as their reductions may be linearly dependent.
Since such reductions are repeated multiple times for different primes~$p$ and evaluation points of~$t$, we use a tracer~\cite{Traverso89} during the computation for the first pair $(p,t)$ to record
all~$\eta'$ corresponding to a non-contributing term and skip the corresponding terms in subsequent computations.
Assuming that the first pair $(p,t)$ is not unlucky, we know that all the skipped pairs would also
lead to unnecessary elements if used in later computations.
The use of this tracer is another reason why our implementation is probabilistic,
as the algorithm may not even terminate if a wrong tracer was computed due to the choice of a first unlucky pair $(p,t)$.
This probability can of course be made arbitrarily small by computing tracers
for multiple pairs and performing a majority test.

\section{Application to the computation of ODEs satisfied by $k$-regular graphs}
\label{sec:kreg}

In this section, we illustrate our new algorithm
with computations on a family of multivariate integrals of combinatorial origin.
We compute linear ODEs satisfied
by various models of $k$-regular graphs and generalizations.
A distinguishing feature of these integration problems is that
they cannot be solved by classical creative telescoping algorithms,
which perform computations over the field of rational functions
in all variables:
because the objects to be integrated have polynomial torsion,
they are not functions,
and such calculations would erroneously result in a zero integral.

ODEs for~$k\leq 5$ were obtained 20~years ago
by naive linear algebra and elimination by Euclidean divisions~\cite{ChyzakMishnaSalvy-2005-ESP}.
This has recently been extended to~${k\leq 7}$
by a multivariate analog~\cite{ChyzakMishna-2024-DES}
of the reduction-based algorithm~\cite{BostanChyzakLairezSalvy-2018-GHR}
in which the reduction is modulo the polynomial image of several differential operators.
Here we use our new algorithm to achieve~$k = 8$.

\subsection{Statement of the integration problem}

We first briefly introduce the problem and its solution by an integral representation.
We refer the reader to~\cite{ChyzakMishnaSalvy-2005-ESP}
for further motivation, history, and details.
Given a fixed integer~$k\geq2$, a $k$-regular graph is
a graph whose vertices all have degree~$k$, that is, all have exactly $k$~neighbors.
We are interested in the enumerative generating function
\begin{equation*}
R_k(t) = \sum_{n\geq0} r_{k,n} \frac{t^n}{n!} ,
\end{equation*}
where $r_{k,n}$~is the number of $k$-regular labeled graphs on $n$~vertices.
A combinatorial result by Gessel~\cite{Gessel-1990-SFP} states
that the generating function~$R_k(t)$~satisfies
a linear ODE with polynomial coefficients in~$t$.
In this section, we show how such a differential equation can be obtained from Algorithm~\ref{algo:integration}.

To this end, we use the classical formulation of~$R_k(t)$
as a scalar product of two exponentials,
\begin{equation*}
R_k(t) = \scalp{e^f}{e^{tg}} ,
\end{equation*}
where $f$ and~$g$ are explicit polynomials in the indeterminates $p_1,\dots,p_k$
\cite[Algorithm~1, Step~a.]{ChyzakMishna-2024-DES},
and where the scalar product is a classic tool
in the combinatorics of symmetric functions.
The scalar product is first defined on monomials~by
\begin{equation}\label{eq:def-scalp}
\scalp{p_1^{r_1}\dotsm p_k^{r_k}}{p_1^{s_1}\dotsm p_k^{s_k}} = z_\rr \delta_{\rr,\ss}
\qquad\text{for}\qquad
z_\rr = r_1!\,1^{r_1} r_2!\,2^{r_2} \dotsm r_k!\,k^{r_k}.
\end{equation}
With $z_\rr$~indexed by exponents in~$\bN^k$,
we depart from the equivalent indexing~$z_\lambda$ by partitions $\lambda_1\geq\lambda_2\geq\dots$
that is used classically as well as in~\cite{ChyzakMishna-2024-DES}.
The scalar product is then extended by bilinearity to
left arguments in~$\bQ[[\pp]]$
and right arguments in~$\bQ[\pp]((t))$,
making the scalar product live in~$\bQ((t))$.
Note that, because in the symmetric-function theory the power function~$p_i$ denotes the sum $x_1^i+x_2^i+\dotsb$,
we use the more traditional~$\pp$ instead of~$\xx$ for the variables, in accordance with existing literature.
Also, by~$\bQ[\pp]((t))$ we mean the ring of formal sums with coefficients in~$\bQ[\pp]$,
with finitely many exponents towards~$-\infty$ and potentially infinitely many towards~$+\infty$.
This is not a field.
The use of Algorithm~\ref{algo:integration} is justified by the following statement,
which we will prove after an example.

\begin{theorem}\label{thm:kreg-main}
  Let~$f, g\in \bQ[\pp]$.
  Let~$S$ be the left ideal of~$W_{\pp}(t)$ generated by
  \[ p_i - t \frac{\partial \tilde g}{\partial X_i}(u_1,\dotsc,u_k) , \qquad 1\leq i\leq k , \]
  where $\tilde g(X_1,\dots,X_k) = g(X_1,2X_2,\dots,kX_k)$
  and $u_j = \frac{\partial f}{\partial p_j} - \partial_j$ for $1\leq j\leq k$.
  Then, $W_{\pp}(t)/S$ is holonomic as a $W_\pp(t)$-module.
  Write~$\proj_S$ for the canonical projection
  $\proj_S\colon W_\pp(t)\rightarrow W_\pp(t)/S$.
  Then, $W_\pp(t)/S$~can be endowed with a derivation~$\partial_t$ commuting with $\pp$ and~$\dd_\pp$ satisfying
  \begin{equation}\label{eq:der-dt}
  \partial_t\cdot\proj_S(a) = \proj_S\left(\tfrac{\partial a}{\partial t} + a \tilde g(u_1,\dotsc,u_k)\right).
  \end{equation}
  Algorithm~\ref{algo:integration},
  applied to the module~$W_{\pp}(t)/S$ (for~$r=1$),
  the derivation~\eqref{eq:der-dt} (for the implied endomorphism $L\colon a\mapsto a \tilde g(u)$),
  the element $1 \in W_\pp(t)$ (denoted~$f$ in the algorithm),
  and any $\rho \geq 0$,
  produces a nonzero differential operator~$P(t, \partial_t)$ such that
  \[ P(t, \partial_t) \cdot \scalp{e^f}{e^{tg}} = 0. \]
\end{theorem}

\begin{example}\label{ex:3-regs}
To compute a differential operator satisfied by the exponential generating function
counting 3-regular graphs, we set $k=3$ and
\begin{equation*}
f = \frac{p_1^2}{2}-\frac{p_2}{2}-\frac{p_2^2}{4}+\frac{p_3^2}{6},
\qquad
g = \frac{p_3}{3}+\frac{p_2p_1}{2}+\frac{p_1^3}{6} .
\end{equation*}
Then, we get $\tilde g(p_1,p_2,p_3) = p_3 + p_2p_1 + p_1^3/6$ and set
\begin{equation*}
u_1 = p_1-\partial_1,
\quad
u_2 = -\frac{1+p_2}{2}-\partial_2,
\quad
u_3 = \frac{p_3}{3}-\partial_3.
\end{equation*}
The~$u_i$ obviously commute with one another,
and this is a general fact beyond this example.
Then, the ideal~$S$ is generated by
\begin{gather*}
p_1 - t \left(u_2 + \frac{u_1^2}{2}\right) = p_1+t+\frac{tp_2}{2}-\frac{tp_1^2}{2}+\frac{tp_1\partial_1}{2}+t\partial_2-\frac{t\partial_1^2}{2} , \\
p_2 - t u_1 = p_2 - t p_1 + t \partial_1 , \qquad p_3 - t .
\end{gather*}
To describe the derivation with respect to~$t$, we obtain
\begin{multline*}
\tilde g(u_1,u_2,u_3) =  u_3 + u_2u_1 + \frac{u_1^3}{6} \\
= \frac{p_1^3}{6}-p_1-\frac{p_1p_2}{2}+\frac{p_3}{3}
+\partial_1+\frac{p_2\partial_1}{2}-\frac{p_1^2\partial_1}{2} \\
-p_1\partial_2-\partial_3
+\frac{p_1\partial_1^2}{2}+\partial_1\partial_2
-\frac{\partial_1^3}{6} .
\end{multline*}
We will not be able to show properties of the map~$\proj_S$ before proving Lemma~\ref{lem:dt-exists}.
Applied to the previous data,
Algorithm~\ref{algo:integration} returns
\begin{align*}
P &= 9t^3(t^4+2t^2-2) \partial_t^2 \\
  &+ 3(t^{10}+6t^8+3t^6-6t^4-26t^2+8) \partial_t - t^3(t^4+2t^2-2)^2 ,
\end{align*}
which annihilates
$1+\frac{t^4}{4!}+\frac{70t^6}{6!}+\frac{19355t^8}{8!}+\frac{11180820t^{10}}{10!}+O(t^{12})$,
the counting series of 3-regular graphs.
\end{example}

The rest of the section is a proof of Theorem~\ref{thm:kreg-main}.
In particular,
we fix the ideal~$S$ as in this statement,
and the holonomicity of~$W_\pp(t)/S$ will be proven as Lemma~\ref{lem:kreg-holonomicity},
the existence of the derivation~$\partial_t$ will be proven as Lemma~\ref{lem:dt-exists},
and Theorem~\ref{thm:algo-integration-correct} will prove the correctness of the output operator~$P$.

We begin with a few preliminary definitions.
First, given $h = h(p_1,\dots,p_k) \in \bQ[\pp]$,
we write $\tilde h$ for $h(1p_1, 2p_2, \dots, kp_k)$.
Second, we define a formal Laplace transform on monomials by
\begin{equation}\label{eq:def-laplace}
\cL(p_1^{r_1} \dots p_k^{r_k}) = r_1! \, p_1^{-r_1-1} \dots r_k! \, p_k^{-r_k-1}
\end{equation}
and extend it by linearity into a map from~$\bQ[\pp]$ to~$\bQ[\pp^{-1}]$.
%
Third, we introduce the ring
\[ \msrp := \bQ[[\pp]][\pp^{-1}] = \bigcup_{\ell\geq0} (p_1\dotsm p_k)^{-\ell} \bQ[[\pp]] \]
and a formal residue on it by the formula
\begin{equation}\label{eq:def-res}
\res \biggl( \sum_{\rr \in \bZ^k} c_\rr \pp^\rr \biggr) = c_{-1,\dots,-1} .
\end{equation}
Lastly, we extend $h\mapsto\tilde h$ to a map from~$\bQ[\pp]((t))$ to itself,
$\cL$ to a map from~$\bQ[\pp]((t))$ to~$\bQ[\pp^{-1}]((t))$,
and $\res$ to a map from~$\msrp((t))$ to~$\bQ((t))$,
by making each of those maps act coefficient-wisely.

The following lemma is the central change of representation that makes us able to appeal to residue calculations.
We exemplify it before its proof.

\begin{lemma}\label{lem:scalp-as-res}
  For any polynomials $f$ and $g$ in $\bQ[\pp]$,
  \[ \scalp{e^f}{e^{tg}} = \res(e^f \cL(e^{t\tilde g})) . \]
\end{lemma}

\begin{example}\label{ex:scalp-as-res}
We explain this result for the simple example $f = p_k^2$, $g = p_k^3$.
Indeed we compute: first,
\begin{equation*}
\scalp{e^{p_k^2}}{e^{tp_k^3}} = \sum_{i,j\geq0} \Bigl\langle \frac{p_k^{2i}}{i!} , \frac{p_k^{3j}}{j!} \Bigr\rangle t^j
= \sum_{\ell\geq0} \Bigl\langle \frac{p_k^{6\ell}}{(3\ell)!} , \frac{p_k^{6\ell}}{(2\ell)!} \Bigr\rangle t^{2\ell}
= \sum_{\ell\geq0} \frac{(6\ell)! \, k^{6\ell}}{(3\ell)! \, (2\ell)!} t^{2\ell} ,
\end{equation*}
where the scalar product can be non-zero only if $2i$ and~$3j$ are equal, thus of the form~$6\ell$;
and second,
\begin{equation*}
\res\bigl(e^{p_k^2} \cL(e^{t(k^3p_k^3)})\bigr) = \sum_{i,j\geq0} \res\biggl(\frac{p_k^{2i}}{i!} \frac{(3j)! \, k^{3j}}{j! \, p_k^{3j+1}}\biggr)
= \sum_{\ell\geq0} \frac{(6\ell)! \, k^{6\ell}}{(3\ell)! \, (2\ell)!} t^{2\ell} ,
\end{equation*}
where again a residue can be non-zero only if $2i = 3j = 6\ell$.
\end{example}

\begin{proof}
For $U \in \bQ[[\pp]]$ and
for $\rr \in \mathbb{N}^k$,
$U \cL(\tilde\pp^\rr)$~is an element of~$\msrp$, so that
successively using \eqref{eq:def-laplace}, \eqref{eq:def-scalp}, \eqref{eq:def-res}, and \eqref{eq:def-scalp} again,
we derive:
\begin{align*}
\res (U \cL(\tilde\pp^\rr))
  &= \res\left(U \, 1^{r_1} \dotsm k^{r_k} r_1! \, p_1^{-r_1-1} \dotsm r_k! \, p_k^{-r_k-1} \right) \\
  &= z_\rr \res\left(U p_1^{-r_1-1} \dotsm p_k^{-r_k-1}\right)
  = z_\rr \, [\pp^\rr] U = \scalp{U}{\pp^\rr} .
\end{align*}
This formula extends by linearity to $\scalp{U}{h} = \res (U \cL(\tilde h))$
for any $h \in \bQ[\pp]$.
Upon specializing to $U = e^f$ and~$h = g^\ell/\ell!$ before taking series in~$t$,
this makes the informal integral formula
provided in~\cite[end of 7.1]{ChyzakMishnaSalvy-2005-ESP} completely algebraic (and formal),
in the form of the formula
\begin{equation*}
\scalp{e^f}{e^{tg}}
= \sum_{\ell\geq0} \scalp{e^f}{g^\ell} \frac{t^\ell}{\ell!}
= \sum_{\ell\geq0} \res\left(e^f \cL(\tilde g^\ell)\right) \frac{t^\ell}{\ell!}
= \res\left(e^f \cL(e^{t\tilde g})\right). \qedhere
\end{equation*}
\end{proof}

The ring~$\msrp((t))$ is a $W_\pp(t)$-module with the usual actions: $p_i$ acts by multiplication and~$\partial_i$ by partial differentiation with respect to~$p_i$.
Let~$K$ be the subspace of all elements of~$\msrp((t))$ that do not contain any monomial~$\pp^\rr t^m$ with $r_1,\dotsc,r_k$ all negative.
In other words,
\[ K = \sum_{i=1}^k \bQ[[\pp]][p_1^{-1},\dotsc,p_{i-1}^{-1},p_{i+1}^{-1},\dotsc,p_k^{-1}]((t)). \]
This subspace has the property to be a sub-$W_\pp(t)$-module of~$\msrp((t))$ and to be contained in the kernel of the residue map~$\msrp((t)) \to \bQ((t))$.

Now, let $f$ and~$g$ be two polynomials in~$\bQ[\pp]$.
We provide in Lemma~\ref{lem:dt-exists} an explicit construction of a derivation~$\partial_t$ satisfying \eqref{eq:der-dt}.
We remark that this construction is simpler than the general approach presented in Section~\ref{subsec:scalar-extension}:
it is inspired by Example~\ref{ex:airy-fin},
but departs from it by already incorporating the change of algebras~$\tau$ that we will use in Lemma~\ref{lem:kreg-holonomicity}.

\begin{lemma}\label{lem:dt-exists}
The $W_\xx(t)$-linear map $L\colon a\mapsto a \, \tilde g\bigl(\tfrac{\partial f}{\partial p_1}-\partial_1,\dotsc,\tfrac{\partial f}{\partial p_k}-\partial_k\bigr)$
defines a derivation on~$W_\xx(t)/S$ by
\begin{equation}\label{eq:kreg-def-dt}
  \partial_t\cdot\proj_S(a) = \proj_S\left(\tfrac{\partial a}{\partial t} + L(a)\right) .
\end{equation}
This derivation commutes with $\pp$ and~$\dd_\pp$.
\end{lemma}

\begin{proof}

Let $\phi$ be the $W_\pp$-linear endomorphism of $W_\pp(t)$ defined by $\phi(a)= \tfrac{\partial a}{\partial t} + L(a)$.
To show that $\partial_t$ is well-defined, it suffices to verify that $\phi(S) \subseteq S$.
To simplify the writing of the present proof, for any polynomial $q(X_1,\dots,X_k)$
we write $q(\uu)$ in place of $q\bigl(\tfrac{\partial f}{\partial p_1}-\partial_1,\dotsc,\tfrac{\partial f}{\partial p_k}-\partial_k\bigr)$.
Consider the generators
\begin{equation}\label{eq:s_i-def}
s_i := p_i - t \frac{\partial \tilde g}{\partial X_i}(\uu)
\end{equation}
of~$S$.
(Here, as we did for the statement of Theorem~\ref{thm:kreg-main},
we write $\partial\tilde g/\partial X_i$ to denote the derivative of~$\tilde g$ with respect to its $i$th variable,
rather than~$\partial\tilde g/\partial p_i$,
to avoid any ambiguities that could arise because of the subsequent substitution of the $i$th variable with a polynomial in $\pp$ and~$\dd_\pp$.)
Using the commutation rule $p_i \bigl(\tfrac{\partial f}{\partial p_j}-\partial_j\bigr) = \bigl(\tfrac{\partial f}{\partial p_j}-\partial_j\bigr) p_i + \delta_{i,j}$,
one obtains for any polynomial $q(X_1,\dots,X_k)$
\begin{equation}\label{eq:p_i-commutation}
p_i \, q(\uu) = q(\uu) \, p_i + \frac{\partial q}{\partial X_i}(\uu) .
\end{equation}
For each~$i$, we therefore get
\begin{equation}\label{eq:phi-s_i}
\phi(s_i) = - \frac{\partial \tilde g}{\partial X_i}(\uu) + s_i \, \tilde g(\uu) = \tilde g(\uu) \, s_i ,
\end{equation}
where the first equality is by the definition of~$\phi$
and the second equality is by the definition~\eqref{eq:s_i-def} and the specialization of~\eqref{eq:p_i-commutation} to~$q = \tilde g$.
Consequently, $s_i \in S$ for each~$i$.
Next, for any rational function~$R(t)$ and any~$i$, we derive
\begin{multline*}
\phi\bigl(R(t) s_i\bigr) = \frac{\partial R}{\partial t}(t) \, s_i
- R(t) \frac{\partial \tilde g}{\partial X_i}(\uu)
+ R(t) s_i \, \tilde g(\uu)
= \frac{\partial R}{\partial t}(t) s_i + R(t) \phi(s_i) ,
\end{multline*}
where the first equality is by \eqref{eq:s_i-def} and the definition of~$\phi$
and the second equality is by~\eqref{eq:phi-s_i}.
Therefore, by $\bQ(t)$-linearity of~$L$,
$\phi\bigl(R(t) s_i\bigr) \in S$ for each~$i$.
Now, $S$~is generated as a $W_\xx$-module by the family~$(R(t)s_i)_{R,i}$, so,
by $W_\xx$-linearity of~$L$, and thus of~$\phi$, we get
the inclusion $\phi(S) \subseteq S$.
The $W_\xx(t)$-linearity of~$L$ and the definition~\eqref{eq:kreg-def-dt} imply, for any~$R(t) \in \bQ(t)$,
\begin{multline*}
\partial_t R(t) \cdot \proj_S(a) =
\partial_t \cdot \proj_S\bigl(R(t) a)\bigr) =
\proj_S\left(\tfrac{\partial(R(t)a)}{\partial t} + R(t) L(a)\right) \\
= \tfrac{\partial R}{\partial t}(t) \proj_S(a) + \proj_S\left(R(t)\bigl(\tfrac{\partial a}{\partial t}+L(a)\bigr)\right)
= R(t) \partial_t \cdot \proj_S(a) + \tfrac{\partial R}{\partial t}(t) \proj_S(a) .
\end{multline*}
In other words, $\partial_t$~is a derivation.
A similar but simpler calculation shows that it commutes with $\pp$ and~$\dd_\pp$.
\end{proof}

For the same polynomials $f$ and~$g$, let~$\Xi_{f,g}$ be the class of $e^f \cL(e^{t\tilde g})$ modulo~$K$.

\begin{lemma}
  For any~$a\in S$, the relation $a \cdot \Xi_{f, g} = 0$ holds.
  Moreover,
  \[
    \tfrac{\partial }{\partial t} \cdot \Xi_{f,g}= \tilde g\bigl(\tfrac{\partial f}{\partial p_1}-\partial_1,\dotsc,  \tfrac{\partial f}{\partial p_k}-\partial_k\bigr) \cdot  \Xi_{f,g}.
  \]
\end{lemma}

\begin{proof}
The definition~\eqref{eq:def-laplace} of the formal Laplace transform
implies the following formulas, valid for~$\rr \in \bN^k$:
\begin{gather}
\label{eq:L-of-p}
\cL(p_i \cdot \pp^\rr) = (r_i+1)p_i^{-1} \cL(\pp^\rr) = -\partial_i \cdot \cL(\pp^\rr) , \\
\label{eq:L-of-dp}
\cL(\partial_i \cdot \pp^\rr) = \cL(r_ip_i^{-1} \pp^\rr) =
  \begin{cases}
    p_i \cdot \cL(\pp^\rr) , & \text{if $r_i \neq 0$} , \\
    0 , & \text{otherwise} .
  \end{cases}
\end{gather}
In turn, for $h \in \bQ[\pp]$, this implies the formulas
\begin{gather}
\label{eq:L-of-p-h}
\cL(p_i \cdot h) = -\partial_i \cdot \cL(h) , \\
\cL(\partial_i \cdot h) = p_i \cdot \cL(h) - p_i \cdot \cL( h\rvert_{p_i=0}) .
\end{gather}
The last formula is not convenient because of the term involving~$h\rvert_{p_i=0}$.
Fortunately, this term is in~$K$, so we have the nicer formula
\begin{equation}
  \label{eq:L-of-dp-h}
  \cL(\partial_i \cdot h) \equiv p_i \cdot \cL(h) \pmod{K}.
\end{equation}
Moreover, we have for any~$h \in \msrp((t))$,
\begin{equation}
  \label{eq:2}
  \bigl(\tfrac{\partial f}{\partial p_i} - \partial_i\bigr) \cdot e^f h = - e^f \partial_i \cdot h .
\end{equation}
Therefore, we have
\begin{align}
  \Bigl( p_i - t &\tfrac{\partial \tilde g}{\partial p_i}\bigl(\tfrac{\partial f}{\partial p_1}-\partial_1,\dotsc,  \tfrac{\partial f}{\partial p_k}-\partial_k\bigr) \Bigr) \cdot e^f \cL(e^{t\tilde g})\label{eq:cancelling-scalp}\\
  &= e^f \Bigl( p_i - t  \tfrac{\partial \tilde g}{\partial p_i}(-\partial_1, \dotsc,-\partial_k) \Bigr) \cdot \cL(e^{t\tilde g}), && \text{using \eqref{eq:2}},\notag\\
                  &\equiv e^f \cL \Bigl( \bigl( \partial_i - t \tfrac{\partial\tilde g}{\partial p_i} \bigr) \cdot e^{t\tilde g} \Bigr) , && \text{using \eqref{eq:L-of-p-h} and \eqref{eq:L-of-dp-h}},\notag\\
                    &\equiv e^f \cL(0) \equiv 0 \pmod{K}.\notag
\end{align}
This proves the first statement about all~$a \in S$ by $W_\pp(t)$-linearity.
The second statement is proved similarly, starting with
$\tfrac{\partial }{\partial t} - \tilde g\bigl(\tfrac{\partial f}{\partial p_1}-\partial_1,\dotsc,  \tfrac{\partial f}{\partial p_k}-\partial_k\bigr)$
instead of the operator in~\eqref{eq:cancelling-scalp}.
\end{proof}

The following lemma introduces a $\bQ(t)$-algebra automorphism~$\tau$ of~$W_\pp(t)$,
so that the ideal~$S$ is generated by the image of~$\dd$ under~$\tau$.
The holonomicity of~$W_\pp(t)/S$ is then a consequence of how $\tau$~transports the Berstein filtration~\eqref{eq:Bernstein-filtration}.

\begin{lemma}\label{lem:kreg-holonomicity}
  The~$W_\pp(t)$-module~$W_\pp(t)/S$ is holonomic.
\end{lemma}

\begin{proof}
  Let~$\tau$ be the endomorphism of the $\bQ(t)$-algebra~$W_\pp(t)$ defined by
  \begin{equation}\label{eq:def-tau}
  \tau(p_i) = \tfrac{\partial f}{\partial p_i} - \partial_i \quad\text{and}\quad \tau(\partial_i) = p_i - t \, \tau \bigl( \tfrac{\partial\tilde g}{\partial p_i} \bigr).
  \end{equation}
  Note that the definition is not recursive since~$\tfrac{\partial\tilde g}{\partial p_i}$ is a polynomial in~$\pp$ only.
  Also, $\tau(p_i)$ and~$\tau(p_j)$ are easily seen to commute with one another,
  so that the images under~$\tau$ of any two elements of~$\bQ(t)$ commute.
  As a consequence, $\tau(\partial_i)\tau(p_i) - \tau(p_i)\tau(\partial_i)$ reduces to~$p_i(-\partial_i) - (-\partial_i)p_i = 1$.
  This justifies that indeed \eqref{eq:def-tau}~properly defines an algebra morphism.
  This is an automorphism as it is easily seen to admit the inverse given by
  \begin{equation}\label{eq:tau-inverse}
  \tau^{-1}(p_i) = \partial_i + t \tfrac{\partial\tilde g}{\partial p_i} \quad\text{and}\quad \tau^{-1}(\partial_i) = \tau^{-1} \bigl( \tfrac{\partial f}{\partial p_i} \bigr) - p_i.
  \end{equation}
  By its definition in Theorem~\ref{thm:kreg-main}, $S$~is generated by~$\tau(\partial_1),\dotsc,\tau(\partial_k)$.
  Therefore, as~$\bQ(t)$-vector spaces,
  \[ W_\pp(t) / S \simeq W_\pp(t) / (W_\pp(t) \partial _1 + \dotsb + W_\pp(t) \partial _k) \simeq \bQ(t)[\pp] , \]
  where the first isomorphism is obtained by applying~$\tau^{-1}$.
  More precisely, $\tau^{-1}$~induces an injective linear map from the filtration $\cF_m\cdot\gamma \subseteq W_\pp(t)/S$,
  where $\gamma$~denotes the class of~$1$ in~$W_\pp(t)/S$,
  to the filtration $\cF_{dm}\cdot 1 \subseteq \bQ(t)[\pp]$,
  where $d$~is the maximal degree of the images~\eqref{eq:tau-inverse} of the generators of~$W_\pp(t)$.
  Since $\bQ(t)[\pp]$~is holonomic, the dimension of~$\cF_{dm}\cdot 1$ is in~$O(m^n)$.
  By injectivity, the dimension of~$\cF_m\cdot\gamma$ is in~$O(m^n)$, too,
  and $W_\pp(t)/S$~is holonomic.
\end{proof}

Let~$M$ denote the holonomic module~$W_\pp(t)/S$.
All in all, we have a commutative diagram of~$\bQ(t)$-linear spaces, with all arrows commuting with the derivation~$\partial_t$
\[
\begin{tikzcd}
M \arrow[r] \arrow[d]  & {W_\pp(t)\cdot \Xi_{f,g}} \arrow[r, hook] \arrow[d, "\res"] & \frac{\msrp((t))}{K} \arrow[d, "\res"] \\
\frac{M}{\dd M} \arrow[r] & W_\pp(t) \cdot \scalp{e^f}{e^{tg}} \arrow[r, hook]      &  \bQ((t)).
\end{tikzcd}
\]
The class~$\gamma$ of~$1$ in~$M$ is mapped to~$\scalp{e^f}{e^{tg}}$ in~$\bQ((t))$.
What Algorithm~\ref{algo:integration} computes, is an operator~$P(t, \partial_t)$ such that~$P(t, \partial_t) \cdot \gamma \in \dd M$.
This implies that $P(t, \partial_t) \cdot \scalp{e^f}{e^{tg}} = 0$.

\begin{remark}
At this point we can make explicit our remark
that earlier creative-telescoping algorithms
could not deal with our integrals.
From the explicit definition of~$g$ in the formula~$R_k(t) = \scalp{e^f}{e^{tg}}$,
we can prove that $\tilde g$~is always in the form $p_k + h(p_1,\dotsc,p_{k-1})$, for some polynomial~$h$.
So $p_k - t$~is always in~$S$.
Any integration algorithm that would work over~$\bQ(t,p_k)$,
as many that are designed to apply to functions,
would therefore consider that $1$~is in the annihilator
of the function to be integrated,
which would lead to a wrong result.
\end{remark}

\subsection{Experimental results}
\label{subsec:experimental-results}

We consider graph models that are either some model of $k$-regular (simple) graphs
or some generalization with loops and/or multiple edges and/or degrees
in the set $\{1, 2, \dots, k\}$ instead of~$\{k\}$.
Given such a graph model, the theory in~\cite{ChyzakMishna-2024-DES}
provides immediate formulas for the polynomials $f$ and~$g$.
Obtaining the ideal $S$ of Theorem~\ref{thm:kreg-main} is easily implemented as a simple non-commutative substitution.
To this end, we used Maple's \verb+OreAlgebra+ package.
After converting\footnote{%
Maple uses its commutative product \texttt{*} to represent monomials,
so that both Maple inputs \texttt{t*dt} and \texttt{dt*t} represent
the element~$t \partial_t$.
This is no problem inside Maple, where non-commutative products are computed
by the command \texttt{skew\_product}.
But naively serializing an operator from Maple by \texttt{lprint} can lead
to strings with a different interpretation in Brochet's Julia implementation.
We automated a rewrite of those strings to move all derivatives to the right.}
from Maple notation to the notation of \verb+MultivariateCreativeTelescoping.jl+,
we could use the latter to obtain the wanted ODEs,
appealing to the implementation of our optimized Algorithm~\ref{algo:integration-crt-e-i}.
This was done for~$2 \leq k \leq 8$
and the degree sets $K = \{k\}$ and $K = \{1, 2, \dots, k\}$.
We collected the computed ODEs and made them available on the web%
\footnote{See \url{https://files.inria.fr/chyzak/kregs/}.}.
For~$k\leq 7$, they are the same as with the Maple implementation
that accompanies Chyzak and Mishna's article~\cite{ChyzakMishna-2024-DES}.
To the best of our knowledge,
the ODEs for~$k=8$ are obtained for the first time.
Table~\ref{table:res} displays some parameters related
to the calculations and to the results we obtained.
There, each graph model is described by a triple~$(k, l, e)$ where:
$k$~determines the set of allowed degrees, either by $K = \{k\}$ for the left part of Table~\ref{table:res} or by $K = \{1,\dots,k\}$ for the right part;
$l$~is one of `ll' for loopless graphs, `la' for graphs with loops allowed and contributing~$2$ to the degree, and `lh' for graphs with loops allowed and contributing~$1$ to the degree;
$e$~is either `se' for graphs with simple edges or `me' for graphs with multiple edges allowed.
For each graph model,
the resulting ODE has order and degree given in columns `ord' and~`deg'.
The total time for computing the ODE is given in column~`total'
and decomposes as follows:
the time to prepare generators for the ideal~$S$ is negligible;
the time to make a Gröbner basis out of them is given in column~`f5';
the time to extract next the ODE by our new algorithm is given in column~`mct'.
We note that the time for~`mct' always dominates.
The total time includes the compilation time, which explains why its value exceeds the sum of `f5' and~`mct'.
The maximal peak of memory usage is listed in column~`rss'.

\begin{table}
  \begin{adjustbox}{width=1.4\textwidth,center}
    \begin{minipage}{0.8\textwidth}
  \begin{tabular}{c c c c c c c c c}
\toprule
\multicolumn{3}{c}{graph} & \multicolumn{2}{c}{ODE} & \multicolumn{4}{c}{time and memory} \\ \cmidrule(r){1-3} \cmidrule(r){4-5} \cmidrule(r){6-9}
$k$ & $l$ & $e$ & ord & deg & f5 & mct & total & rss \\
\midrule
2 & ll & se & 1 & 2 & {\small 0.04} & {\small 0.05 }& 18 & 0.63  \\
3 & ll & se & 2 & 11 & {\small 0.04} & {\small 0.05} & 18 & 0.62  \\
4 & ll & se & 2 & 14 & {\small 0.05} & {\small 0.05} & 19 & 0.62 \\
5 & ll & se & 6 & 125 & {\small 0.05} & {\small 0.59} & 17 & 0.65  \\
6 & ll & se & 6 & 145 & {\small 0.23} & {\small 1.0} & 20 & 0.66 \\
7 & ll & se & 20 & 1683 & {\small 8.6} & {\small 303} & 330 & 4.6  \\
7 & ll & me & 20 & 1683 & {\small 8.4} & {\small 300} & 326 & 4.3 \\
7 & la & se & 20 & 1683 & {\small 8.3} & {\small 301} & 328 & 4.3 \\
7 & la & me & 20 & 1683 & {\small 8.3} & {\small 299} & 326 & 4.4  \\
7 & lh & se & 20 & 1683 & {\small 8.6} & {\small 310} & 337 & 5.9 \\
7 & lh & me & 20 & 1683 & {\small 8.3} & {\small 322} & 349 & 5.8 \\
8 & ll & se & 19 & 1793 & {\small 244} & {\small 832} & 1095 & 6.5 \\
8 & ll & me & 19 & 1793 & {\small 247} & {\small 831} & 1097 & 6.7 \\
8 & la & se & 19 & 1793 & {\small 244} & {\small 831} & 1094 & 6.0 \\
8 & la & me & 19 & 1793 & {\small 244} & {\small 829} & 1093 & 6.0 \\
8 & lh & se & 35 & 6204 & {\small 393} & {\small 23069} & 23481 & 5.4 \\
8 & lh & me & 35 & 6200 & {\small 393} & {\small 23111} & 23524 & 5.4 \\
\bottomrule
\end{tabular}

\end{minipage}
\begin{minipage}{0.75\textwidth}
  \begin{tabular}{c c c c c c c c c}
    \toprule
    \multicolumn{3}{c}{graph} & \multicolumn{2}{c}{ODE} & \multicolumn{4}{c}{time and memory} \\ \cmidrule(r){1-3} \cmidrule(r){4-5} \cmidrule(r){6-9}
    $k$ & $l$ & $e$ & ord & deg & f5 & mct & total & rss \\
    \midrule
        2 & ll & se & 1 & 3 & {\small 0.04} & {\small 0.04} & 17 & 0.62 \\
        3 & ll & se & 2 & 11 & {\small 0.04} & {\small 0.13} & 18 & 0.64 \\
        4 & ll & se & 3 & 29 & {\small 0.04} & {\small 0.07} & 18 & 0.62 \\
        5 & ll & se & 6 & 125 & {\small 0.06} & {\small 0.69} & 19 & 0.65 \\
        6 & ll & se & 10 & 425 & {\small 0.31} & {\small 11} & 31 & 0.86 \\
        7 & ll & se & 20 & 1683 &{\small  8.3} & {\small 316} & 343 & 5.8 \\
        7 & ll & me & 20 & 1683 & {\small 8.5} & {\small 321} & 348 & 5.8 \\
        7 & la & se & 20 & 1683 & {\small 8.7} & {\small 324} & 352 & 5.8 \\
        7 & la & me & 20 & 1683 & {\small 8.8} & {\small 322} & 349 & 5.8 \\
        7 & lh & se & 20 & 1683 & {\small 8.7} & {\small 310} & 337 & 5.6 \\
        7 & lh & me & 20 & 1683 & {\small 8.2} & {\small 323} & 349 & 5.8 \\
        8 & ll & se & 35 & 6201 & {\small 389} & {\small 23198} & 23605 & 5.4 \\
        8 & ll & me & 35 & 6200 & {\small 386} & {\small 23586} & 23991 & 5.4 \\
        8 & la & se & 35 & 6204 & {\small 401} & {\small 23495} & 23915 & 5.4 \\
        8 & la & me & 35 & 6205 & {\small 393} & {\small 23188} & 23600 & 5.4 \\
        8 & lh & se & 35 & 6205 & {\small 387} & {\small 22745} & 23152 & 5.5 \\
        8 & lh & me & 35 & 6205 & {\small 394} & {\small 23440} & 23853 & 5.4 \\
    \bottomrule
\end{tabular}

\end{minipage}
\end{adjustbox}
\cprotect\caption{\label{table:res}%
Computation of ODEs for $k$-regular graph models $K = \{k\}$ (left) and $K = \{1,\dots,k\}$ (right).
A~graph model is input by the triple $(k,l,e)$,
where $l$ and~$e$ describe if loops and edges are allowed.
For each output ODE, the order (`ord') and the coefficient degree (`deg') are given.
All times are given in seconds.
The two main computation steps are preparing a Gröbner basis (`f5')
and running Algorithm~\ref{algo:integration-crt-e-i} on it (`mct').
The maximum memory used is given in GB (`rss').
See details in Section~\ref{subsec:experimental-results}.}
\end{table}

By way of comparison, the model (se,~ll,~7) could be computed
by the method and implementation of~\cite{ChyzakMishna-2024-DES}
in $3.22\cdot10^4$~seconds (almost 9~hours),
which is roughly 100~times as much as the 330~seconds (5.5~minutes)
needed by our new algorithm and implementation.
This can partly be explained
by the lack of efficient evaluation/interpolation methods
in the implementation of~\cite{ChyzakMishna-2024-DES}
and by the choice in~\cite{ChyzakMishna-2024-DES}
to continue the calculation
by factoring the polynomial coefficients of the ODE.

The case~$k=8$ (and $K = \{8\}$) shows a peculiar behavior regarding the order of the ODEs found.
Up to~$k = 7 = \max K$, for numerous models that we have investigated (see more data in~\cite{ChyzakMishna-2024-DES}),
the dependency of the order `ord' and the degree `deg' is monotonous in~$k$.
For some sets~$K$, but not systematically, the order or the degree of the ODE
is larger when~$l=\text{lh}$.
Focusing on~$k = 8$, Table~\ref{table:res} shows that
ODEs for~$K = \{k\}$ have
smaller order than for~$k=7$ if~$l\neq\text{la}$
and larger order than for~$k=7$ if~$l=\text{la}$,
but that
ODEs for~$K = \{1,\dots,k\}$ have
the same typical (larger) order for~$k=8$ than for~$k=7$, whatever choice is made for~$l$.
On the other hand, the scalar size of the operators, measured as~$(\text{`ord'}+1)\times(\text{`deg'}+1)$,
increases with~$k$.
We have no complete explanation for the deviation of order,
beyond the observation that it could be related to the parity of~$k$.
Indeed, the number of edges in a $k$-regular graph or multigraph on $n$~vertices is~$kn/2$,
so that the generating functions for $K=\{k\}$ and~$l\neq\text{la}$ are even series for odd~$k$,
while all other generating functions have no parity structure.
Accordingly, the operators for the even series are (up to a minor renormalization)
skew polynomials in $t^2$ and~$t \partial_t$,
while in other cases they show no special structure.
The orders (and degrees) increase in parallel for odd~$k$ and for even~$k$.
To get some added certainty that the drop of the order is not caused by any implementation mistake,
we implemented a fully independent counting method by backtracking
so as to obtain the numbers of 8-regular graphs on $n$~vertices for~$n \leq 48$:
we could confirm that the ODE for~$K = \{8\}$ annihilates the obtained truncated series up to~$O(t^{49})$.
We also ran Maple's \verb+mindiffeq+ command on the operator obtained for $(k,l,e) = (7,\text{`ll'},\text{`se'})$,
which confirmed that it is minimal among annihilators of the series.

\medskip
\noindent{\bf Acknowledgements.}
This work has been supported by the European Research Council under the European Union's
Horizon Europe research and innovation programme, grant agreement
101040794 (10000 DIGITS).

\printbibliography

@Book{AdamsLoustaunau-1994-IGB,
  author    = {Adams, William W. and Loustaunau, Philippe},
  publisher = {American Mathematical Society},
  title     = {An Introduction to {G}röbner Bases},
  year      = {1994},
}

@InCollection{Bahloul-2020-GBD,
  author    = {Bahloul, Rouchdi},
  booktitle = {Two Algebraic Byways from Differential Equations: {G}rö\-bner Bases and Quivers},
  publisher = {Springer},
  title     = {Gröbner bases in {$D$}-modules: application to {B}ernstein-{S}ato polynomials},
  year      = {2020},
  isbn      = {978-3-030-26453-6; 978-3-030-26454-3},
  pages     = {75--93},
  series    = {Algorithms Comput. Math.},
  volume    = {28},
  doi       = {10.1007/978-3-030-26454-3_2},
}

@Book{BeckerWeispfenning-1993-GB,
  author    = {Becker, Thomas and Weispfenning, Volker},
  publisher = {Springer},
  title     = {Gröbner {B}ases},
  year      = {1993},
  doi       = {10.1007/978-1-4612-0913-3},
}

@Article{Bernstein-1971-MRD,
  author  = {Bernshtein, I. N.},
  journal = {Func. Anal. Appl.},
  title   = {Modules over a ring of differential operators: an investigation of the fundamental solutions of equations with constant coefficients},
  year    = {1971},
  issn    = {0374-1990},
  note    = {Transl. from \emph{Akademija Nauk SSSR. Funkcional'nyi Analiz i ego Prilozhenija}},
  number  = {2},
  pages   = {1--16},
  volume  = {5},
  doi     = {10.1007/BF01076413},
}

@Article{Beukers-1983-IPP,
  author    = {Beukers, Frits},
  journal   = {Progr. Math.},
  title     = {Irrationality of $\pi^2$, periods of an elliptic curve and $\Gamma_1(5)$},
  note      = {Proc. of Diophantine Approximations and Transcendental Numbers (Luminy, 1982)},
  pages     = {47--66},
  volume    = {31},
  date      = {1983},
  publisher = {Birkhäuser},
}

@Article{bezanson2017,
  author       = {Bezanson, Jeff and Edelman, Alan and Karpinski, Stefan and Shah, Viral B.},
  title        = {Julia: a fresh approach to numerical computing},
  number       = {1},
  pages        = {65--98},
  volume       = {59},
  date         = {2017},
  doi          = {10.1137/141000671},
  journaltitle = {SIAM Rev.},
  publisher    = {SIAM},
}

@Book{Bjork-1979-RDO,
  author    = {Björk, J.-E.},
  publisher = {North-Holland},
  title     = {Rings of Differential Operators},
  isbn      = {0-444-85292-1},
  series    = {North-Holland Mathematical Library},
  volume    = {21},
  date      = {1979},
}

@Book{Borel-1987-ADM,
  author    = {Borel, A. and Grivel, P.-P. and Kaup, B. and Haefliger, A. and Malgrange, B. and Ehlers, F.},
  publisher = {Academic Press},
  title     = {Algebraic {$D$}-Modules},
  isbn      = {0-12-117740-8},
  series    = {Perspectives in Mathematics},
  volume    = {2},
  date      = {1987},
}

@InProceedings{BostanChenChyzakLi_2010,
  author    = {Bostan, Alin and Chen, Shaoshi and Chyzak, Frédéric and Li, Ziming},
  booktitle = {Proc. 35th ISSAC},
  title     = {Complexity of creative telescoping for bivariate rational functions},
  pages     = {203--210},
  publisher = {ACM},
  date      = {2010},
  doi       = {10.1145/1837934.1837975},
  isbn      = {978-1-4503-0150-3},
}

@InProceedings{BostanChenChyzakLiXin_2013,
  author    = {Bostan, Alin and Chen, Shaoshi and Chyzak, Frédéric and Li, Ziming and Xin, Guoce},
  booktitle = {Proc. 38th ISSAC},
  title     = {Hermite reduction and creative telescoping for hyperexponential functions},
  pages     = {77--84},
  publisher = {ACM},
  date      = {2013},
  doi       = {10.1145/2465506.2465946},
  isbn      = {978-1-4503-2059-7},
}

@InProceedings{BostanChyzakLairezSalvy-2018-GHR,
  author    = {Bostan, Alin and Chyzak, Frédéric and Lairez, Pierre and Salvy, Bruno},
  booktitle = {ISSAC'18},
  title     = {Generalized {H}ermite reduction, creative telescoping and definite integration of {D}-finite functions},
  editor    = {Schost, Éric},
  pages     = {95--102},
  publisher = {ACM},
  date      = {2018},
  doi       = {10.1145/3208976.3208992},
}

@InProceedings{BostanLairezSalvy_2013,
  author    = {Bostan, Alin and Lairez, Pierre and Salvy, Bruno},
  booktitle = {Proc. ISSAC '13},
  title     = {Creative telescoping for rational functions using the {G}riffiths–{D}work method},
  pages     = {93--100},
  publisher = {ACM},
  date      = {2013},
  doi       = {10.1145/2465506.2465935},
  isbn      = {978-1-4503-2059-7},
}

@Thesis{Castro-1984-TDO,
  author      = {Castro, Francisco},
  date        = {1984},
  institution = {Université Paris VII},
  title       = {Théorème de division pour les opérateurs différentiels et calcul des multiplicités},
  type        = {Thèse de $3^{\rm e}$ cycle},
}

@InCollection{Castro-1987-CEI,
  author    = {Castro, F.},
  booktitle = {Géométrie algébrique et applications. {III} Géométrie réelle. Systèmes différentiels et théorie de {H}odge},
  publisher = {Hermann, Paris},
  title     = {Calculs effectifs pour les idéaux d'opérateurs différentiels},
  isbn      = {2-7056-6029-1},
  note      = {Proc. of Deuxième conférence internationale de {L}a {R}ábida (1984)},
  pages     = {1--19},
  series    = {Travaux en Cours},
  volume    = {24},
  date      = {1987},
  doi       = {10.1007/BF00191372},
}

@InProceedings{ChenDuKauers_2023,
  author    = {Chen, Shaoshi and Du, Lixin and Kauers, Manuel},
  booktitle = {Proc. ISSAC '23},
  title     = {Hermite reduction for {D}-finite functions via integral bases},
  pages     = {155--163},
  publisher = {ACM},
  date      = {2023},
  doi       = {10.1145/3597066.3597092},
  isbn      = {979-8-4007-0039-2},
}

@Article{ChenHoeijKauersKoutschan_2018,
  author       = {Chen, Shaoshi and van Hoeij, Mark and Kauers, Manuel and Koutschan, Christoph},
  title        = {Reduction-based creative telescoping for {F}uchsian {D}-finite functions},
  issn         = {0747-7171},
  pages        = {108--127},
  volume       = {85},
  date         = {2018},
  doi          = {10.1016/J.JSC.2017.07.005},
  journaltitle = {J. Symbolic Comput.},
}

@InProceedings{ChenHuangKauersLi_2015,
  author    = {Chen, Shaoshi and Huang, Hui and Kauers, Manuel and Li, Ziming},
  booktitle = {Proc. 40th ISSAC},
  title     = {A modified {A}bramov-{P}etkov\v sek reduction and creative telescoping for hypergeometric terms},
  editor    = {{ACM}},
  date      = {2015},
  doi       = {10.1145/2755996.2756648},
}

@InProceedings{ChenKauersKoutschan_2016,
  author    = {Chen, Shaoshi and Kauers, Manuel and Koutschan, Christoph},
  booktitle = {Proc. ISSAC '16},
  title     = {Reduction-based creative telescoping for algebraic functions},
  pages     = {175--182},
  publisher = {ACM},
  date      = {2016},
  doi       = {10.1145/2930889.2930901},
  isbn      = {978-1-4503-4380-0},
}

@Article{Chyzak_2000,
  author       = {Chyzak, Frédéric},
  title        = {An extension of {Z}eilberger's fast algorithm to general holonomic functions},
  issn         = {0012-365X},
  number       = {1--3},
  pages        = {115--134},
  volume       = {217},
  date         = {2000},
  doi          = {10.1016/S0012-365X(99)00259-9},
  journaltitle = {Discrete Math.},
}

@Thesis{Chyzak-2014-ACT,
  author      = {Chyzak, Frédéric},
  date        = {2014},
  institution = {Université d'Orsay},
  note        = {64 pages},
  title       = {The {ABC} of Creative Telescoping: Algorithms, Bounds, Complexity},
  type        = {Memoir of accreditation to supervise research ({HDR})},
}

@Unpublished{ChyzakMishna-2024-DES,
  author = {Chyzak, Frédéric and Mishna, Marni},
  note   = {To appear in \emph{Combinatorial Theory}},
  title  = {Differential equations satisfied by generating functions of 5-, 6-, and 7-regular labelled graphs: a reduction-based approach},
  date   = {2025},
  urlhal = {https://inria.hal.science/hal-04604501}
}

@Article{ChyzakMishnaSalvy-2005-ESP,
  author       = {Chyzak, Frédéric and Mishna, Marni and Salvy, Bruno},
  title        = {Effective scalar products of {D}-finite symmetric functions},
  number       = {1},
  pages        = {1--43},
  volume       = {112},
  date         = {2005},
  journaltitle = {J. Combin. Theory Ser.~A},
}

@Article{CHYZAK1998187,
  author       = {Chyzak, F. and Salvy, B.},
  title        = {Non-commutative elimination in {O}re algebras proves multivariate identities},
  issn         = {0747-7171},
  number       = {2},
  pages        = {187--227},
  volume       = {26},
  date         = {1998},
  doi          = {10.1006/JSCO.1998.0207},
  journaltitle = {J. Symbolic Comput.},
  url          = {https://www.sciencedirect.com/science/article/pii/S0747717198902073},
}

@Book{Coutinho_1995,
  author    = {Coutinho, S. C.},
  publisher = {Cambridge University Press},
  title     = {A Primer of Algebraic D-Modules},
  year      = {1995},
  isbn      = {9780511623653},
  doi       = {10.1017/CBO9780511623653},
}

@Book{CoxLittleOShea-1998-UAG,
  author    = {Cox, David and Little, John and O'Shea, Donal},
  publisher = {Springer},
  title     = {Using Algebraic Geometry},
  isbn      = {0-387-98487-9; 0-387-98492-5},
  series    = {Graduate Texts in Mathematics},
  volume    = {185},
  date      = {1998},
  doi       = {10.1007/978-1-4757-6911-1},
}

@Article{Dimca_1990,
  author       = {Dimca, Alexandru},
  title        = {On the {M}ilnor fibrations of weighted homogeneous polynomials},
  issn         = {0010-437X},
  number       = {1--2},
  pages        = {19--47},
  volume       = {76},
  date         = {1990},
  journaltitle = {Compositio Math.},
  url          = {http://www.numdam.org/item?id=CM_1990__76_1-2_19_0},
}

@Article{Dwork62,
  author       = {Dwork, Bernard},
  title        = {On the zeta function of a hypersurface},
  pages        = {5--68},
  volume       = {12},
  date         = {1962},
  journaltitle = {Publ. Math. Inst. Hautes Études Sci.},
  publisher    = {Institut des Hautes Études Scientifiques},
  url          = {http://www.numdam.org/item/PMIHES_1962__12__5_0/},
}

@Article{Dwork64,
  author       = {Dwork, Bernard},
  title        = {On the zeta function of a hypersurface: {II}},
  issn         = {0003486X},
  number       = {2},
  pages        = {227--299},
  volume       = {80},
  date         = {1964},
  journaltitle = {Ann. Math.},
  url          = {http://www.jstor.org/stable/1970392},
}

@Article{FAUGERE199961,
  author       = {Faugère, Jean-Charles},
  title        = {A new efficient algorithm for computing {G}röbner bases ({F4})},
  issn         = {0022-4049},
  number       = {1},
  pages        = {61--88},
  volume       = {139},
  date         = {1999},
  doi          = {10.1016/S0022-4049(99)00005-5},
  journaltitle = {J. Pure Appl. Algebra},
  url          = {https://www.sciencedirect.com/science/article/pii/S0022404999000055},
}

@InCollection{Galligo-1985-SAQ,
  author    = {Galligo, André},
  booktitle = {Eurocal'85, {V}ol.\ 2 ({L}inz, 1985)},
  publisher = {Springer},
  title     = {Some algorithmic questions on ideals of differential operators},
  pages     = {413--421},
  series    = {Lecture Notes in Comput. Sci.},
  volume    = {204},
  date      = {1985},
}

@Book{Gathen1999,
  author    = {von zur Gathen, Joachim and Gerhard, Jürgen},
  publisher = {Cambridge University Press},
  title     = {{M}odern {C}omputer {A}lgebra},
  isbn      = {0-521-64176-4},
  citekey   = {GathenGerhard-1999-MCA},
  date      = {1999},
}

@Article{Gessel-1990-SFP,
  author   = {Gessel, Ira M.},
  journal  = {J. Combin. Theory Ser.~A},
  title    = {Symmetric functions and {P}-recursiveness},
  year     = {1990},
  number   = {2},
  pages    = {257--285},
  volume   = {53},
  doi      = {10.1016/0097-3165(90)90060-A},
  mrnumber = {1041448},
}

@Article{Griffiths-1969-PCR,
  author       = {Griffiths, Phillip A.},
  title        = {On the periods of certain rational integrals: {I}},
  number       = {3},
  pages        = {460--495},
  volume       = {90},
  date         = {1969},
  doi          = {10.2307/1970746},
  journaltitle = {Ann. Math.},
}

@Article{Hoeven_2021,
  author       = {van der Hoeven, Joris},
  title        = {Constructing reductions for creative telescoping},
  issn         = {1432-0622},
  number       = {5},
  pages        = {575--602},
  volume       = {32},
  date         = {2021},
  doi          = {10.1007/S00200-020-00413-3},
  journaltitle = {Appl. Algebra Engrg. Comm. Comput.},
}

@Article{KandriRodyWeispfenning-1990-NGB,
  author       = {Kandri-Rody, A. and Weispfenning, V.},
  title        = {Noncommutative {G}röbner bases in algebras of solvable type},
  issn         = {0747-7171},
  number       = {1},
  pages        = {1--26},
  volume       = {9},
  date         = {1990},
  journaltitle = {J. Symbolic Comput.},
}

@InBook{KauersJaroschekJohansson-2015-OPS,
  author    = {Kauers, Manuel and Jaroschek, Maximilian and Johansson, Fredrik},
  editor    = {Gutierrez, Jaime and Schicho, Josef and Weimann, Martin},
  pages     = {105--125},
  publisher = {Springer},
  title     = {Ore polynomials in {S}age},
  isbn      = {978-3-319-15081-9},
  booktitle = {Computer Algebra and Polynomials: Applications of Algebra and Number Theory},
  date      = {2015},
  doi       = {10.1007/978-3-319-15081-9_6},
}

@Report{Koutschan_2010,
  author      = {Koutschan, Christoph},
  date        = {2010},
  institution = {RISC Report Series, University of Linz, Austria},
  number      = {10--01},
  title       = {{HolonomicFunctions} user's guide},
}

@Article{Koutschan_2010a,
  author       = {Koutschan, Christoph},
  title        = {A fast approach to creative telescoping},
  number       = {2--3},
  pages        = {259--266},
  volume       = {4},
  date         = {2010},
  doi          = {10.1007/S11786-010-0055-0},
  journaltitle = {Math. Comput. Sci.},
}

@Article{Lairez-2016-CPR,
  author       = {Lairez, Pierre},
  title        = {Computing periods of rational integrals},
  issn         = {0025-5718},
  number       = {300},
  pages        = {1719--1752},
  volume       = {85},
  date         = {2016},
  doi          = {10.1090/MCOM/3054},
  journaltitle = {Math. Comp.},
  url          = {http://dx.doi.org/10.1090/mcom/3054},
}

@Article{LAIREZ2024102275,
  author       = {Lairez, Pierre},
  title        = {Axioms for a theory of signature bases},
  issn         = {0747-7171},
  pages        = {102275},
  volume       = {123},
  date         = {2024},
  doi          = {10.1016/J.JSC.2023.102275},
  journaltitle = {J. Symbolic Comput.},
  url          = {https://www.sciencedirect.com/science/article/pii/S0747717123000895},
}

@Article{Laporta_2000,
  author       = {Laporta, S.},
  title        = {High-precision calculation of multiloop {F}eynman integrals by difference equations},
  issn         = {0217-751X},
  number       = {32},
  pages        = {5087--5159},
  volume       = {15},
  date         = {2000},
  doi          = {10.1142/S0217751X00002159},
  journaltitle = {Int. J. Mod. Phys. A},
  publisher    = {World Scientific Publishing Co.},
}

@InProceedings{LevandovskyySchonemann-2003-PCA,
  author    = {Levandovskyy, Viktor and Schönemann, Hans},
  booktitle = {Proc. ISSAC '03},
  title     = {Plural: a computer algebra system for noncommutative polynomial algebras},
  publisher = {ACM},
  date      = {2003},
  doi       = {10.1145/860854.860895},
}

@Article{Oaku_2013,
  author       = {Oaku, Toshinori},
  title        = {Algorithms for integrals of holonomic functions over domains defined by polynomial inequalities},
  issn         = {0747-7171},
  pages        = {1--27},
  volume       = {50},
  date         = {2013},
  doi          = {10.1016/J.JSC.2012.05.004},
  journaltitle = {J. Symbolic Comput.},
}

@Article{OakuTakayama_2001,
  author       = {Oaku, Toshinori and Takayama, Nobuki},
  title        = {Algorithms for {D}-modules: restriction, tensor product, localization, and local cohomology groups},
  issn         = {0022-4049},
  number       = {2},
  pages        = {267--308},
  volume       = {156},
  date         = {2001},
  doi          = {10.1016/S0022-4049(00)00004-9},
  journaltitle = {J. Pure Appl. Algebra},
}

@Article{OakuTakayamaWalther-2000-LAD,
  author       = {Oaku, Toshinori and Takayama, Nobuki and Walther, Uli},
  title        = {A localization algorithm for {$D$}-modules},
  issn         = {0747-7171},
  note         = {Special issue on Symbolic Computation in Algebra, Analysis, and Geometry (Berkeley, CA, 1998)},
  number       = {4-5},
  pages        = {721--728},
  volume       = {29},
  date         = {2000},
  doi          = {10.1006/JSCO.1999.0398},
  journaltitle = {J. Symbolic Comput.},
  url          = {http://dx.doi.org/10.1006/jsco.1999.0398},
}

@Book{prudnikovIntegralsSeries31986,
  author    = {Prudnikov, Anatolii Platonovich and Brychkov, Iurii Alexandrovich and Marichev, Oleg Igorevich and Gould, G. G.},
  publisher = {Gordon and Breach},
  title     = {Integrals and Series 3: More Special Functions},
  year      = {1986},
  isbn      = {2-88124-682-6},
}

@Book{SaitoSturmfelsTakayama-2000-GDH,
  author    = {Saito, Mutsumi and Sturmfels, Bernd and Takayama, Nobuki},
  publisher = {Springer},
  title     = {Gröbner Deformations of Hypergeometric Differential Equations},
  isbn      = {3-540-66065-8},
  series    = {Algorithms and Computation in Mathematics},
  volume    = {6},
  date      = {2000},
}

@InProceedings{10.1145/2442829.2442879,
  author    = {Sun, Yao and Wang, Dingkang and Ma, Xiaodong and Zhang, Yang},
  booktitle = {Proc. 37th ISSAC},
  title     = {A signature-based algorithm for computing Gröbner bases in solvable polynomial algebras},
  pages     = {351--358},
  publisher = {ACM},
  date      = {2012},
  doi       = {10.1145/2442829.2442879},
  isbn      = {9781450312691},
  url       = {https://doi.org/10.1145/2442829.2442879},
}

@Article{Takayama-1989-GBP,
  author       = {Takayama, Nobuki},
  title        = {Gröbner basis and the problem of contiguous relations},
  issn         = {0910-2043},
  number       = {1},
  pages        = {147--160},
  volume       = {6},
  date         = {1989},
  doi          = {10.1007/BF03167920},
  journaltitle = {Japan Journal of Applied Mathematics},
  url          = {http://dx.doi.org/10.1007/BF03167920},
}

@InProceedings{Takayama_1990,
  author    = {Takayama, Nobuki},
  booktitle = {Proc. ISSAC '90},
  title     = {Gröbner basis, integration and transcendental functions},
  pages     = {152--156},
  publisher = {ACM},
  date      = {1990},
  doi       = {10.1145/96877.96916},
  isbn      = {978-0-201-54892-1},
}

@InProceedings{Takayama-1990-ACI,
  author    = {Takayama, Nobuki},
  booktitle = {Proc. ISSAC '90},
  title     = {An algorithm of constructing the integral of a module: an infinite dimensional analog of {G}röbner basis},
  pages     = {206--211},
  publisher = {ACM},
  date      = {1990},
}

@InProceedings{Traverso89,
  author    = {Traverso, Carlo},
  booktitle = {Symbolic and Algebraic Computation},
  title     = {Gröbner trace algorithms},
  editor    = {Gianni, P.},
  pages     = {125--138},
  publisher = {Springer},
  date      = {1989},
  isbn      = {978-3-540-46153-1},
}

@Article{WilfZeilberger_1992,
  author       = {Wilf, Herbert S. and Zeilberger, Doron},
  title        = {An algorithmic proof theory for hypergeometric (ordinary and “$q$”) multisum/integral identities},
  issn         = {0020-9910},
  number       = {3},
  pages        = {575--633},
  volume       = {108},
  date         = {1992},
  doi          = {10.1007/BF02100618},
  journaltitle = {Invent. Math.},
}

@Article{Zeilberger_1990a,
  author       = {Zeilberger, Doron},
  title        = {A fast algorithm for proving terminating hypergeometric identities},
  issn         = {0012-365X},
  number       = {2},
  pages        = {207--211},
  volume       = {80},
  date         = {1990},
  doi          = {10.1016/0012-365X(90)90120-7},
  journaltitle = {Discrete Math.},
  langid       = {english},
}

@article{staffordModuleStructureWeyl1978,
  title = {Module structure of Weyl algebras},
  author = {Stafford, J. T.},
  year = 1978,
  journal = {Journal of the London Mathematical Society},
  volume = {s2-18},
  number = {3},
  pages = {429--442},
  issn = {1469-7750},
  doi = {10.1112/JLMS/S2-18.3.429},
  copyright = {\copyright{} 1978 London Mathematical Society},
  langid = {english}
}

@Article{Strassen-1973-VD,
  author   = {Strassen, Volker},
  journal  = {J. Reine Angew. Math.},
  title    = {Vermeidung von {D}ivisionen},
  year     = {1973},
  pages    = {184--202},
  volume   = {264},
  doi      = {10.1515/crll.1973.264.184},
}

@InBook{Zippel-1979-PAS,
  author    = {Zippel, Richard},
  pages     = {216--226},
  publisher = {Springer},
  title     = {Probabilistic algorithms for sparse polynomials},
  year      = {1979},
  isbn      = {9783540351283},
  booktitle = {Symbolic and Algebraic Computation},
  doi       = {10.1007/3-540-09519-5_73},
}

@Article{Kaltofen-1988-GCD,
  author     = {Kaltofen, Erich},
  journal    = {J. Assoc. Comput. Mach.},
  title      = {Greatest common divisors of polynomials given by straight-line programs},
  year       = {1988},
  number     = {1},
  pages      = {231--264},
  volume     = {35},
  doi        = {10.1145/42267.45069},
}

@Article{KaltofenTrager-1990-CPG,
  author   = {Kaltofen, Erich and Trager, Barry M.},
  journal  = {J. Symbolic Comput.},
  title    = {Computing with polynomials given by black boxes for their evaluations: greatest common divisors, factorization, separation of numerators and denominators},
  year     = {1990},
  number   = {3},
  pages    = {301--320},
  volume   = {9},
  doi      = {10.1016/S0747-7171(08)80015-6},
}

\end{document}